\pdfoutput=1
\documentclass[%
 reprint,
superscriptaddress,
 amsmath,amssymb,
 aps,
 pra,
]{revtex4-2}

\usepackage{caption}
\usepackage{subcaption}
\usepackage{graphicx}
\usepackage{amsmath}
\usepackage{geometry}
\usepackage{hyperref}
\usepackage[T1]{fontenc}
\usepackage[utf8]{inputenc}
\usepackage{float}
\usepackage{csquotes}
\usepackage[version=4]{mhchem}
\usepackage{amssymb}
\usepackage{xcolor}
\usepackage{bbold}
\usepackage{dcolumn}
\usepackage{bm}
\usepackage{algorithm}
\usepackage{algpseudocode}

\usepackage{amsthm}
\newtheorem{theorem}{Theorem}
\newtheorem{lemma}{Lemma}

\newcommand{\beginsupplement}{%
        \setcounter{table}{0}
        \renewcommand{\thetable}{S\arabic{table}}%
        \setcounter{figure}{0}
        \renewcommand{\thefigure}{S\arabic{figure}}%
        \setcounter{equation}{0}
        \renewcommand{\theequation}{S\arabic{equation}}%
        \setcounter{section}{0}
        \renewcommand{\thesection}{S\Roman{section}}%
     }

\begin{document}
\preprint{}

\title{
Exact Realization of Deep ReLU Computation in a Microscopic Statistical-Mechanical System
}

\author{Junxu Li}
\email{lijunxu1996@gmail.com}
\affiliation{Department of Physics, College of Science, Northeastern University, Shenyang 110819, China}
\date{\today}

\begin{abstract}
Understanding the physical origin of the piecewise-linear computation in deep rectified linear unit (ReLU) networks remains a fundamental challenge.
Here we establish an exact statistical-mechanical realization of arbitrary ReLU networks. 
Starting from a microscopic configuration space and its state multiplicities, we construct a partition function without prescribing a neural-network activation function.
The resulting system admits equivalent descriptions in terms of cascaded quantum operations with post selection, or fermionic transport model.
We rigorously prove that in $\beta\to+\infty$ limit, the thermodynamic observables of this system exactly reproduce the hidden states, outputs, and loss function of an arbitrary deep ReLU network.
Crucially, we demonstrate that the piecewise-linear behavior of the ReLU network emerges from a first-order phase transition in the microscopic system, where the discontinuities precisely coincide with the boundaries of linear regions of ReLU network.
This work establishes an exact physical foundation for deep neural computation and provides a novel statistical-mechanical perspective on neural architecture design.
\end{abstract}

\maketitle

\section{Introduction}
Artificial neural networks have become a central paradigm for computation and learning\cite{lecun2015deep, schmidhuber2015deep, goodfellow2016deep}, yet the microscopic principles underlying their remarkably effective computational rules remain poorly understood\cite{sejnowski2020unreasonable, saxe2021if}. 
In particular, the rectified linear unit (ReLU) has a deceptively simple form\cite{nair2010rectified, glorot2011deep},
\begin{equation}
\operatorname{ReLU}(u)=\max(0,u),
\end{equation}
but its repeated application through deep networks generates a highly structured partition of parameter and input space into distinct linear regions\cite{pascanu2013number,montufar2014number}.
This piecewise-linear structure is central to the representation and optimization of modern deep networks, but its physical origin remains unclear\cite{ montufar2014number, serra2018bounding}.
A natural question is therefore whether ReLU computation can arise as the emergent behavior of a microscopic physical system governed by standard statistical-mechanical principles\cite{laurent2018multilinear, he2020piecewise}.

Statistical mechanics has long provided a useful language for neural networks, linking learning dynamics\cite{hopfield1982neural, amit1985storing, watkin1993statistical, bahri2020statistical}, generalization, collective behavior, and phase transitions to thermodynamic descriptions \cite{baity2018comparing,geiger2019jamming,geiger2020scaling,ghio2024sampling}.
Much of this literature starts from a prescribed neural-network architecture
 nd uses statistical-mechanical methods to characterize its collective or learning behavior \cite{li2021statistical,pacelli2023statistical}, including mean-field
theories of neural-network activations \cite{peterson1987mean, mei2019mean, sirignano2020mean} and studies of phase structure in finite- and infinite-width networks \cite{oostwal2021hidden,luo2021phase}.
In the complementary direction, recent studies have shown that physical principles can be used to construct computational systems with neural network-like behavior and, in some cases, to generate effective neural nonlinearities \cite{yoshioka2019transforming, wright2022deep, abbasi2024physical}.
These developments nevertheless leave open an important question, can a
statistical-mechanical system specified independently at the microscopic
level realize an arbitrary deep ReLU computation exactly, including its
layerwise hidden states, output, and loss?

Here we answer this question affirmatively. 
We construct a statistical-mechanical system from its microscopic configuration space and state multiplicities and derive the partition function, without imposing any ReLU activation rule. 
The same construction has two complementary physical realizations, cascaded quantum operations with post selection\cite{aaronson2005quantum, aliferis2007accuracy} and a fermionic transport model.
These interpretations provide a microscopic picture in which trajectories propagate through successive layers, and lead to a common set of thermodynamic observables whose behavior can be analyzed for arbitrary network architectures.

We rigorously prove an exact correspondence between this physical system and deep ReLU networks in the zero-temperature limit $\beta\to+\infty$.
Specifically, we show that the thermodynamic observables of the system, the population imbalance, its output, and the associated loss, exactly reproduce the hidden states, outputs, and loss function of a given ReLU network.
These identities hold for arbitrary inputs, network parameters, depths, and layer widths. 
Furthermore, we demonstrate that the piecewise-linear structure of deep ReLU networks\cite{laurent2018multilinear, he2020relu, tao2022piecewise} emerges from a first-order phase transition in the microscopic system, where the exhibits discontinuities coincide exactly with the boundaries of the linear regions.
Our results provide a microscopic physical foundation for ReLU networks.
Beyond offering a new interpretation of neural computation, this correspondence opens the possibility of studying neural-network structure using the tools of statistical mechanics and of designing computational architectures directly from physical principles.

\section{Implementation of affine transformations}
\label{Sec1}

\begin{figure}[t]
    \begin{center}
        \includegraphics[width=0.48\textwidth]{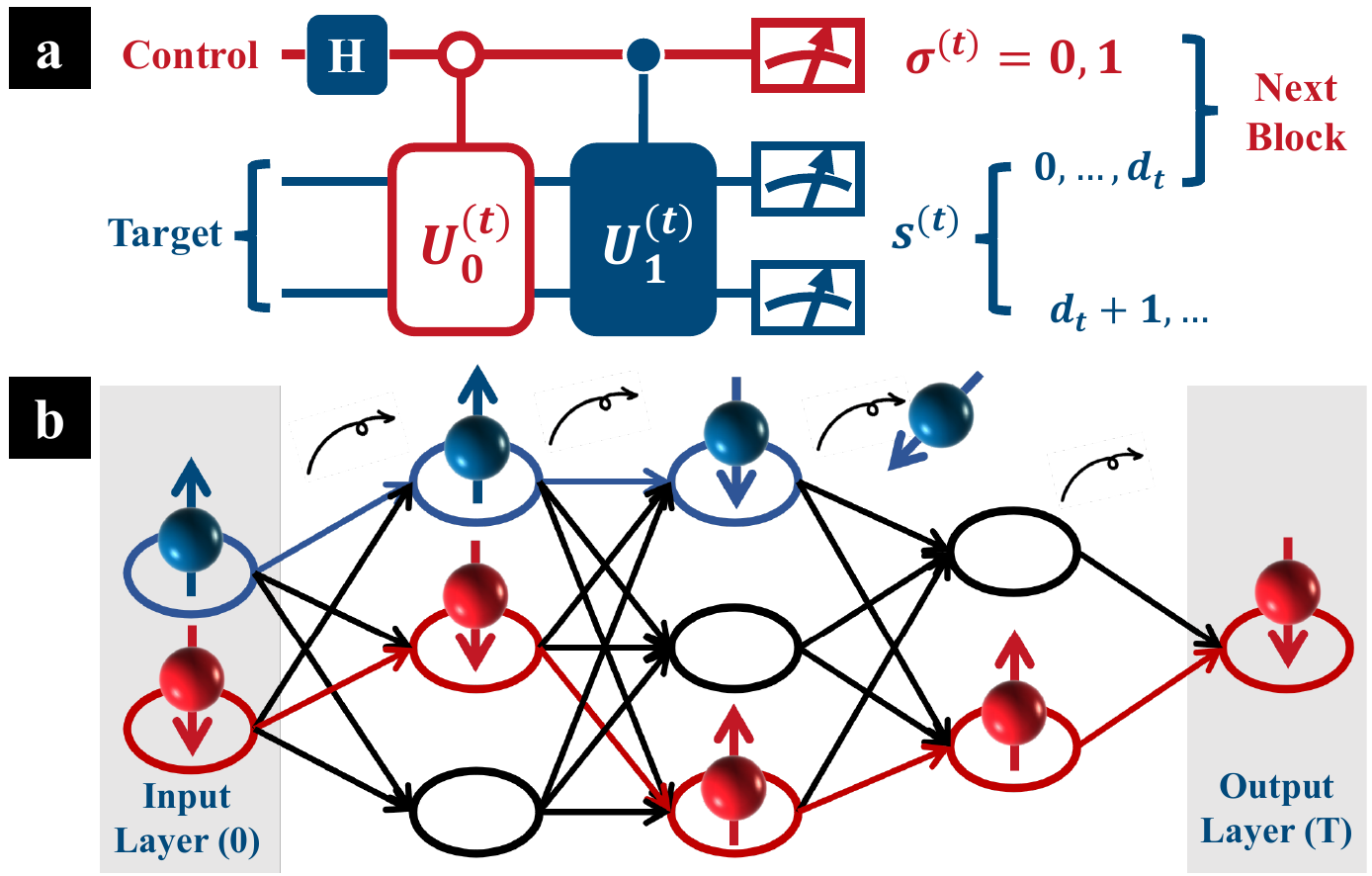}
    \end{center}
    \caption{
{\bf Physical realization of an affine transformation and its fermionic transport interpretation.}
(a) A quantum circuit block implementing the affine map $\widetilde{W}^{(t)}$. 
After measurement, the control and target outcomes are recorded as $\sigma^{(t)}\in\{0,1\}$ and $s^{(t)}$, respectively.
Shots with $s^{(t)}\in\{1,\ldots,d_t\}$ are passed to the next block, whereas shots with $s^{(t)}>d_t$ are discarded.
(b) Representation as a fermionic transport model.
The $t$th column contains $d_t$ effective sites, and each spin-$1/2$ fermion represents one shot.
The fermion propagates from one column to the next, eg, from site $s^{t-1}$ in the $t-1$-th column to the $s^{t}$ site in the next, where spin-up and spin-down
states correspond to $\sigma^{(t)}=0$ and $1$.
Fermions that successfully reaches the $T$-th layer correspond to all $s^{(t)}\in\{1,\ldots,d_t\}$(red), whereas escaped Fermions correspond to certain $s^{(t)}>d_t$(blue).
}
    \label{fig1}
\end{figure}

Consider an arbitrary $T$-layer ReLU network with input $\hat{x}$, the hidden state of the $t$-th layer is given as
\begin{equation}
\mathbf{h}^{(t)}
=
\operatorname{ReLU}\!\left(
W^{(t)}\mathbf{h}^{(t-1)}+\mathbf{b}^{(t)}
\right),
\quad t=1,\ldots,T,
\end{equation}
where $W^{(t)}\in\mathbb{R}^{D_t\times D_{t-1}}$ are weight matrices and
$\mathbf{b}^{(t)}\in\mathbb{R}^{D_t}$ are bias. 
For simplicity, $\mathbf{h}^{(0)}=\hat{\mathbf{x}}$ is the original input, whereas $\mathbf{h}^{(T)}$ is the output.
The affine transformation $W^{(t)}\mathbf{h}^{(t-1)}+\mathbf{b}^{(t)}$ can be converted as a linear transformation $\widetilde W^{(t)}\widetilde{\mathbf h}^{(t-1)}$, where we define
\begin{equation}
\widetilde{\mathbf{h}}^{(t)}
=
\begin{pmatrix}
\mathbf{h}^{(t)}\\
1
\end{pmatrix},
\qquad
\widetilde W^{(t)}
=
\begin{pmatrix}
W^{(t)} & \mathbf{b}^{(t)}\\
\mathbf{0}^{\mathsf T} & 1
\end{pmatrix}
\end{equation}
We have $\widetilde W^{(t)}\in\mathbb{R}^{d_t\times d_{t-1}}$, $\widetilde{\mathbf h}^{(t)}\in\mathbb{R}^{d_t}$, where $d_t=D_t+1$.
Especially, $\widetilde{\mathbf h}^{(0)}=(\hat{\mathbf{x}}^T, 1)^T$.

The real matrix $\widetilde W^{(t)}$ can be embedded into a block with controlled operations as depicted in Fig.(\ref{fig1}a).
The physical implementation of $\widetilde W^{(t)}$, in brief, is as follows.
First, split $\widetilde W^{(t)}$ as
$\widetilde W^{(t)}=\widetilde W^{(t)}_+-\widetilde W^{(t)}_-$,
ensuring that each component of $\widetilde W^{(t)}_{\pm}$ is non-negative.
Next, rescale $\widetilde W^{(t)}_{\pm}$ by $C^{(t)}\geq\max\left\{\left\|\widetilde{W}^{(t)}_{\pm}\right\|_2\right\}$, yielding $\widetilde W^{(t)}_{\pm} = \frac12 {C^{(t)}}\omega^{(t)}_{\pm}$.
$\widetilde\omega^{(t)}_{\pm}$ are both contractions, and thus can be embedded into a unitary operator by dilation.
By appropriately setting $C^{(t)}$ and $U_{0,1}$, we have
\begin{equation}
    \begin{split}
    \widetilde W^{(t)}_{kj}
    &=
    \frac12 {C^{(t)}}
    \left(
        |\langle k|U^{(t)}_0|j\rangle|^2-|\langle k|U^{(t)}_1|j\rangle|^2
    \right)
    \\
    k&=1,\ldots,d_t,\qquad j=1,\ldots,d_{t-1}
    \end{split}
    \label{eq_wtkj}
\end{equation}
For simplicity, we denote $\widetilde \omega^{(t)}=\left(\widetilde \omega^{(t)}_{+} - \widetilde \omega^{(t)}_{-}\right)/2$, where $\left(\widetilde\omega^{(t)}_{\pm}\right)_{kj}=|\langle k|U_{0,1}|j\rangle|^2$.

The input shots of the $t$-th block are all prepared at the computational basis, with control qubit at $\sigma^{(t-1)}=0,1$, and target qubit at state $s^{(t-1)}=1,\ldots, d_{t-1}$.
After applying the operations as depicted in Fig.(\ref{fig1}a), all qubits are measured.
Record the readout of control qubit as $\sigma^{(t)}=0,1$, and target qubit as $s^{(t)}=1,\ldots, d_{t},d_{t}+1,\ldots$.
The shots with $s^{(t)}=1,\ldots, d_{t}$ are sent to the next block as input, as they correspond to the transformation of $\widetilde W^{(t)}$, whereas the other shots are discard.
Denote the number of shots after applying the $t$-th block, with readout $\sigma^{(t)}=0,1$, $s^{(t)}=k$ as $N^{(t)}_{0,k}$ and $N^{(t)}_{1,k}$, respectively, and denote the total number of shots as $N$, initially.
We denote population imbalance ${\mathbf z}^{(t)}\in\mathbb{R}^{d_t}$, with the $k$-th component as
\begin{equation}
    z^{(t)}_k=\frac{N^{(t)}_{0,k}-N^{(t)}_{1,k}}{N}, \qquad k=1,\ldots,d_t
\end{equation}
$z^{(k)}_k$ describes the difference of shots with same target qubit readout $k$, but with different control qubit readouts.
Especially, we denote $\mathbf{x}=\widetilde{\mathbf h}^{(0)}/C^{(0)}$ as the rescaled inputs, where $C^{(0)}\geq \sum_{j=1}^{d_0}|h^{(0)}_j|$.
For simplicity, we denote ${\mathbf z}^{(0)}={\mathbf x}$.
The $t$-th block therefore realizes linear transformation
\begin{equation}
{\mathbf z}^{(t)}
=
\widetilde \omega^{(t)}
{\mathbf z}^{(t-1)}
\end{equation}
After measurement, the state is diagonal in the computational basis and has the same form required as the input of the next block.
Thus the cascade of $t$ blocks implements
\begin{equation}
    \begin{split}
    {\mathbf z}^{(t)}
    &=\widetilde \omega^{(t)}
    \widetilde \omega^{(t-1)}
    \cdots
    \widetilde \omega^{(1)}
    {\mathbf z}^{(0)}
    \\&=
    \left(
        \prod_{\tau=1}^{t}C^{(\tau)}
    \right)^{-1}
    \widetilde W^{(t)}
    \widetilde W^{(t-1)}
    \cdots
    \widetilde W^{(1)}
    {\mathbf z}^{(0)}.
    \end{split}
    \label{eq_cascade}
\end{equation}
\section{Fermionic transport model interpretation}
The cascade of the proposed quantum circuit blocks can be interpreted as a directed, layered fermionic transport network.
Consider a two-dimensional network consisting of $T$ columns as shown in Fig.(\ref{fig1}b), where the $t$-th column contains $d_t$ sites and corresponds to the target-register states $s^{(t)}\in\{1,\ldots,d_t\}$ of the $t$-th circuit block.
The columns are ordered from left to right and represent successive stages of the quantum circuit.
Initially, an ensemble of fermions is distributed over the sites in the leftmost column.
Each spin-1/2 fermion represents one measurement shot, and its initial position and spin are determined by the corresponding input state.
At each step, a fermion propagates from the $t$-th column to the $(t+1)$-th column according to the transition amplitudes generated by the corresponding quantum circuit block.
Fermions may either propagate to a site in the next column (corresponding to $s^{(t)}=1,\ldots,d_t$, which are sent to the next block as input) or escape from the network (corresponding to $s^{(t)}=d_t+1,\ldots$, which are discard).

For the fermions that remain in the network, we record their complete trajectories
\begin{equation}
    \boldsymbol{\sigma}
    =
    \left(
    \sigma^{(0)},\ldots,\sigma^{(T)}
    \right),
    \qquad
    \mathbf{s}
    =
    \left(
    s^{(0)},\ldots,s^{(T)}
    \right),
\end{equation}
where $s^{(t)}\in\{1,\ldots,d_t\}$ records the readout of target qubits, and $\sigma^{(t)}\in\{0,1\}$ specifies the state of the control qubit.
$\sigma^{(t)}$ is identified with the two spin states, $\uparrow\downarrow$, as $(-1)^{\sigma^{(t)}}=\pm1$.
These complete trajectories are `effective', as they correspond to the full affine transformations in the original ReLU network.
Recalling Eq.(\ref{eq_wtkj},\ref{eq_cascade}), the probability of observing a complete trajectory ${\sigma},\mathbf{s}$ is denoted by
\begin{equation}
    \begin{split}
    p(\boldsymbol{\sigma},\mathbf{s})
    =
    \prod_{\tau=1}^T
    \frac12 |\langle s^{(\tau)}|U_{\sigma^{(\tau)}}|s^{(\tau-1)}\rangle|^2
    \end{split}
\end{equation}
Notice that $p(\boldsymbol{\sigma},\mathbf{s})$ is probability of the effective trajectory itself, the probability of choosing certain initial state $\sigma^{(0)},s^{(0)}$ is not included.

\section{Partition function of the proposed model}
\label{Partition function of the proposed model}
We then introduce post-selection after the measurements of each block.
Focusing on the effective trajectories, we derive the partition function for all possible post-selection criteria.

For simplicity, we define indicators $\alpha^{(t)}_{s^{(t)}}\in\{0,1\}$.
After the $t$-th block, shots with readout $\sigma^{(t)}\in\{0,1\},s^{(t)}\in\{1,\ldots,d_t\}$ will be sent to the next block as inputs if $\alpha^{(t)}_{s^{(t)}}=1$, otherwise, if $\alpha^{(t)}_{s^{(t)}}=0$, these shots will be blocked.
In the fermionic transport model interpretation, $\alpha^{(t)}_{s^{(t)}}\in\{0,1\}$ describes if the site $s^{(t)}$ in the $t$-th layer is open (1) or closed (0).
If a site is open, all shots corresponding to that measurement outcome are retained and passed to the next circuit block as input.
If a site is closed, the corresponding shots are blocked at this site, and do not propagate further.
For simplicity, we collect the post-selection criteria (or open/closed configuration) of all sites as a single vector
$\boldsymbol{\alpha}
    =
    \left(
    \alpha^{(1)}_1,\ldots,\alpha^{(1)}_{d_1},
    \ldots,
    \alpha^{(T)}_1,\ldots,\alpha^{(T)}_{d_T}
    \right)$
with which, the population imbalance in Eq.(\ref{eq_cascade}) is rewritten as
\begin{equation}
    \mathbf{\widetilde z}^{(t)}(\boldsymbol{\alpha})
    = \hat{\alpha}^{(t)}\widetilde \omega^{(t)}
    \hat{\alpha}^{(t-1)}\widetilde \omega^{(t-1)}
    \cdots
    \hat{\alpha}^{(1)}\widetilde \omega^{(1)}
    {\mathbf z}^{(0)}
    \label{eq_ztalpha}
\end{equation}
where $\hat{\alpha}^{(t)}=diag\{\alpha^{(t)}_1, \ldots, \alpha^{(t)}_{d_t}\}$.

Consider an ensemble of all possible post-selection criteria (or open/closed configuration) $\boldsymbol{\alpha}$.
Let $q(\boldsymbol{\alpha})$ denote the probability of these configurations. 
Assuming that the inputs and the circuit blocks $\{W^{(t)}\}_{t=1}^T$ are fixed. The Shannon entropy due to uncertainty of $\boldsymbol{\alpha}$ is given by
\begin{equation}
    S = -\sum_{\boldsymbol{\alpha}} q(\boldsymbol{\alpha}) \ln q(\boldsymbol{\alpha}).
\end{equation}
Hereafter, we focus on the effective shots, $\sigma^{(t)}\in\{0,1\},s^{(t)}\in\{1,\ldots,d_t\}$, which corresponds to affine transformations of the original ReLU network, and fermions that always stay in the network.
However, the overall population imbalance $\mathbf{\widetilde z}^{(t)}$ as given in Eq.(\ref{eq_ztalpha}) contains trivial shots, with possible $s^{(\tau)}>d_\tau$ in the successive $T-t$ blocks $\tau=t+1,\ldots,T$.
To address this issue, we define the effective population imbalance ${\pi}^{(t)}$ of all sites $s^{(t)}=1,\ldots,d_t$ of the $t$-th layer,
\begin{equation}
    \pi^{(t)}(\boldsymbol{\alpha}) = 
    \sum_{\boldsymbol{\sigma},\mathbf{s}}
    p(\boldsymbol{\sigma},\mathbf{s})x_{s_0}
    \prod_{\tau=1}^{t}(-1)^{\sigma^{(\tau)}}
    \alpha^{(\tau)}_{s^{(\tau)}}
\end{equation}
where the sum runs over all effective trajectories that pass by $s^{(t)}=k$.
In fermionic transport model, $\pi^{(t)}$ describes the average spin of sites $s^{(t)}=1,\ldots,d_t$ in $t$-th layer over all effective trajectories.

At equilibrium, the entropy $S$ is maximized subject to the constraints $\delta \left(q(\boldsymbol{\alpha})\sum_{t=1}^T2^{T-t}\pi^{(t)}(\boldsymbol{\alpha})\right) = 0$.
The factor $2^{T-t}$ accounts for the degeneracy arising from the fact that we only track the control states $\sigma$ for the first $t$ blocks, while the control states for the remaining $T-t$ blocks remain unknown (each contributing a factor of 2).
By applying Lagrange multipliers, we consider the functional
\begin{equation}
    \delta S + \delta \left(q(\boldsymbol{\alpha})\sum_{t=1}^T2^{T-t}\pi^{(t)}(\boldsymbol{\alpha})\right) = 0,
\end{equation}
where $\beta$ are the associated multipliers.
Substituting the expression for $S$, and noticing that the variation $\delta q(\boldsymbol{\alpha})$ is arbitrary, we find the equilibrium distribution
\begin{widetext}
    \begin{equation}
     q(\boldsymbol{\alpha}) 
     = \frac{1}{Z}\exp \left(\beta\sum_{t=1}^T  2^{T-t}\pi^{(t)}(\boldsymbol{\alpha})\right)
     = \frac{1}{Z}\exp \left(\beta\sum_{t=1}^T  2^{T-t} \sum_{\boldsymbol{\sigma}, \mathbf{s}} p(\boldsymbol{\sigma}, \mathbf{s}) x_{s_0} \prod_{\tau=1}^t (-1)^{\sigma^{(\tau)}} \alpha^{(\tau)}_{s^{(\tau)}} \right)
     \label{eq_qalpha}
\end{equation}
\end{widetext}
with corresponding partition function $Z$ as
\begin{equation}
    Z = \sum_{\boldsymbol{\alpha}} \exp \left(\beta\sum_{t=1}^T  2^{T-t}\pi^{(t)}(\boldsymbol{\alpha})\right).
    \label{eq_partition}
\end{equation}
For simplicity, we define $\Pi(\boldsymbol{\alpha})=\sum_{t=1}^T2^{T-t}\pi^{(t)}(\boldsymbol{\alpha})$, which serves like an effective potential.
Further details on the derivation of $Z$ are provided in the Supplemental Material.

\begin{figure*}[t]
    \begin{center}
        \includegraphics[width=0.9\textwidth]{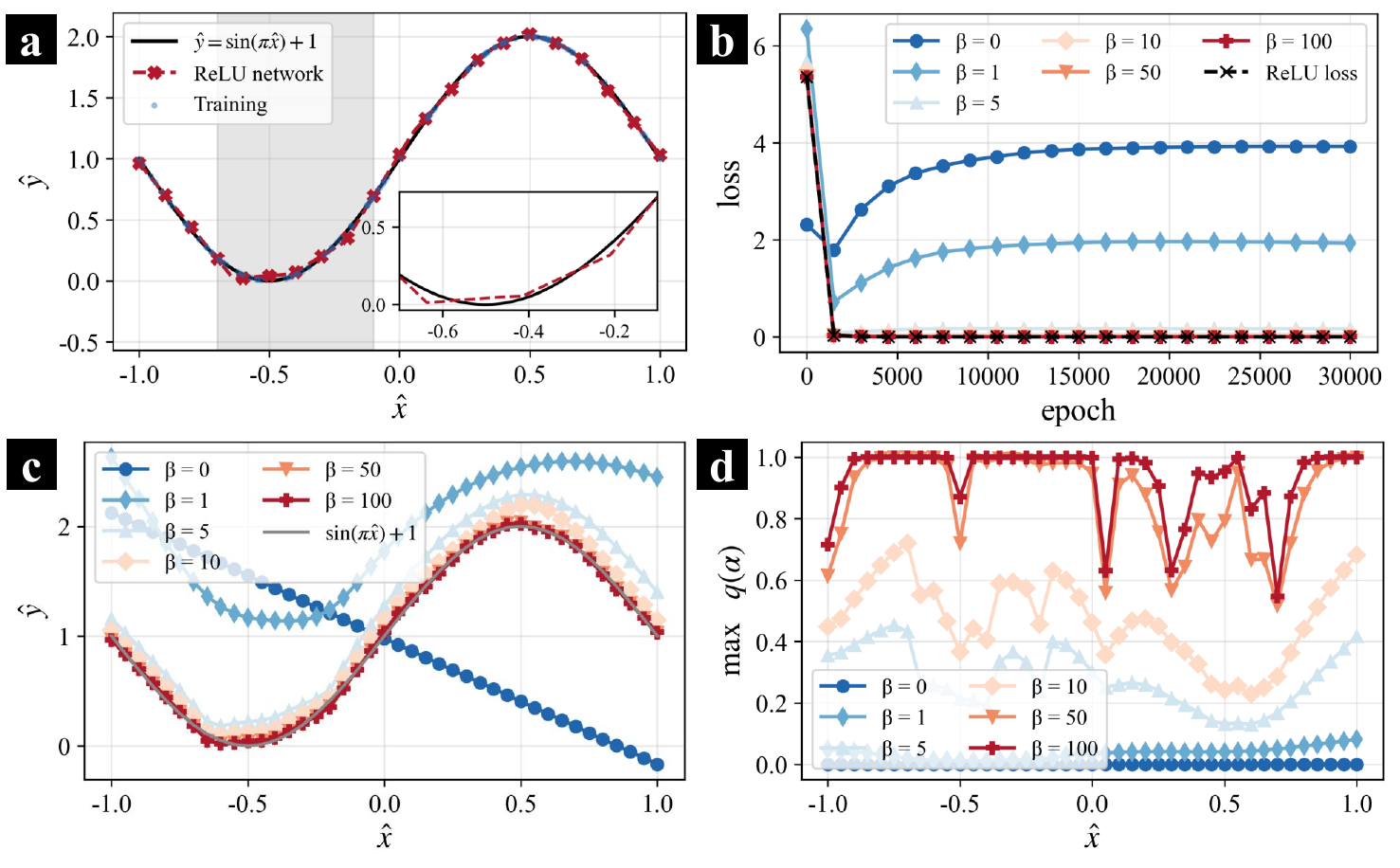}
    \end{center}
    \caption{\textbf{Exact mapping and temperature-driven condensation.}
(a) Prediction of the trained ReLU network (red dashed line) versus the target $\hat{y}=\sin(\pi \hat{x})+1$ (black solid line). Inset: zoom-in showing the piecewise linearity.
(b) Comparison of the training loss (black dashed line) and the ensemble loss (solid lines) across different $\beta$.
(c) The rescaled output $(\prod_{\tau=0}^{T} C^{(\tau)})g$ for various $\beta$, showing the transition from a disordered linear state ($\beta=0$) to the exact ReLU output ($\beta=100$).
(d) Maximum probability $\max q(\boldsymbol{\alpha})$ versus input $\hat{x}$, demonstrating the condensation to the $\widetilde{\boldsymbol{\alpha}}$.}
    \label{fig2}
\end{figure*}

\section{Correspondence with deep ReLU network}
\label{Correspondence with deep ReLU network}
At equilibrium, we obtain the expected value of the population imbalance as
\begin{equation}
    \begin{split}
        \langle\mathbf{\widetilde z}^{(t)}\rangle
    =&\sum_{\boldsymbol{\alpha}}q(\boldsymbol{\alpha})\mathbf{\widetilde z}^{(t)}(\boldsymbol{\alpha})
    =\frac{1}{Z}\sum_{\boldsymbol{\alpha}}\mathbf{\widetilde z}^{(t)}(\boldsymbol{\alpha})\exp \left(\beta\Pi(\boldsymbol{\alpha})\right)
    \end{split}
\end{equation}
where $\langle\cdot\rangle$ indicates the average across all possible ${\boldsymbol{\alpha}}$.
$\langle\mathbf{\widetilde z}^{(t)}\rangle$ corresponds to the hidden states $\widetilde{\mathbf h}^{(t)}$ in the ReLU network.
Especially, we denote the expectation value of the $k$-th component of ${\widetilde z}^{(T)}$ as
\begin{equation}
    g_k= \langle\widetilde z_k^{(T)}(\boldsymbol{\alpha})\rangle
    =\sum_{\boldsymbol{\alpha}}q(\boldsymbol{\alpha})z_k^{(T)}(\boldsymbol{\alpha})
\end{equation}
$g_k$ corresponds to the $k$-th component of the ReLU network output.
Hereafter, we will show the exact correspondence between these observables and the  given deep ReLU network, in the zero-temperature limit, $\beta\to\infty$.
For clarity, the original input is denoted as $\hat{\mathbf{x}}$, which is rescaled as $\mathbf{x}$ so that can be embedded into a mixed state.
The output of the deep ReLU network as $\mathbf{f}(\hat{\mathbf{x}})$, with its $k$-th component as ${f}_k(\hat{\mathbf{x}})$, and the predicted output (training dataset) is denoted as $\hat{\mathbf{y}}$.
\begin{theorem}
At $\beta\to+\infty$ limit, with input $\mathbf{x}$, the ensemble as described by Eq.(\ref{eq_ztalpha}, \ref{eq_partition}) condenses into the unique configuration $\boldsymbol{\alpha}$ that exactly matches the activation pattern of the corresponding deep ReLU network. 
Consequently, the ensemble average of the observables becomes an exact rescaling of the ReLU network hidden states and outputs, for arbitrary depth and width, we have
\begin{equation}
    \lim_{\beta\to+\infty} \langle\mathbf{\widetilde z}^{(t)}\rangle = \left(\prod_{\tau=0}^tC^{(\tau)}\right)^{-1}\widetilde{\mathbf h}^{(t)}
    \label{eq_exacth}
\end{equation}
\begin{equation}
    \lim_{\beta\to+\infty} \left\langle\sum_{k=1}^{D_T}\left(\hat{y}_k-\left(\prod_{\tau=0}^tC^{(\tau)}\right)g_k\right)^2\right\rangle
    =\sum_{k=1}^{D_T}(\hat{y}_k-f_k)^2
    \label{eq_exact_loss}
\end{equation}
\label{theorem1}
\end{theorem}

\begin{proof}
Denote $\widetilde{\boldsymbol{\alpha}}=\arg\max\limits_{\boldsymbol{\alpha}}\Pi(\boldsymbol{\alpha})$.
At $\beta\to+\infty$ limit, only the post-selection criteria $\widetilde{\boldsymbol{\alpha}}$ survives, yielding that $\lim_{\beta\to+\infty}q(\widetilde{\boldsymbol{\alpha}})=1$.
The core of this proof is to find $\widetilde{\boldsymbol{\alpha}}$.

Define a series $A^{(t)}_{s^{(t)}}$ for each $s^{(t)}=1,\ldots,d_t$.
We set $A^{(T)}_{s^{T}}=1$, and iteratively, for $t=2,3,\ldots,T$,
\begin{equation}
    \begin{split}
        A^{(t-1)}_{s^{(t-1)}}=
        &\sum_{s^{(t)}}\alpha^{(t-1)}_{s^{(t-1)}}
\tilde{\omega}_{s^{(t-2)}\to s^{(t-1)}}^{(t-1)}A^{(t)}_{s^{(t)}}
\\&+
2^{T-t+1}\sum_{s^{(t)}}\cdots\sum_{s^{(T)}}\prod_{\tau=t}^T{\Omega}_{s^{(\tau-1)}\to s^{(\tau)}}^{(\tau)}
    \end{split}
\label{eq_At_itera}
\end{equation}
where $\Omega^{(\tau)}=\widetilde\omega_+ + \widetilde\omega_-$, and $\Omega_{s^{(\tau-1)}\to s^{(\tau)}}^{(\tau)}=\Omega_{s^{(\tau)},s^{(\tau-1)}}^{(\tau)}$, we rewrite the subscripts to express the propagation direction of the network.
For $t=1,2,\ldots,T-1$, $s^{(t)}=1,2\ldots,d_t$, we have 
$A^{(t)}_{s^{(t)}}\geq \sum_{s^{(t+1)}}\cdots\sum_{s^{(T)}}\prod_{\tau=t+1}^T{\Omega}_{s^{(\tau-1)}\to s^{(\tau)}}^{(\tau)} \geq0$ (The explicit proof of this lemma are provided in the Supplemental Material).
We can simplify $q(\boldsymbol{\alpha})$ with these series.
Substituting $A^{(1)}_{s^{(1)}}$, we have
\begin{equation}
    q(\boldsymbol{\alpha})=\prod_{s^{(1)}}\exp\left(\beta
    A^{(1)}_{s^{(1)}}\alpha^{(1)}_{s^{(1)}}
    \sum_{s^{(0)}}
        \tilde{\omega}_{s^{(0)}\to s^{(1)}}^{(1)}
        x_{s^{(0)}}
    \right)
    \label{eq_qA1}
\end{equation}
Thus, $\widetilde{\boldsymbol{\alpha}}$ must maximize each terms in the product.
Recalling that $\alpha^{(1)}_{s^{(1)}}\in\{0,1\}$, we have
\begin{equation}
    \widetilde\alpha^{(1)}_{s^{(1)}}=
    \begin{cases} 
  1, & \sum_{s^{(0)}}
        \tilde{\omega}_{s^{(0)}\to s^{(1)}}^{(1)}
        x_{s^{(0)}}\geq0 \\ 
  0,      & \sum_{s^{(0)}}
        \tilde{\omega}_{s^{(0)}\to s^{(1)}}^{(1)}
        x_{s^{(0)}}<0
\end{cases}
\label{eq_talpha1}
\end{equation}
which further yielding that
\begin{equation}
   \begin{split}
       &\widetilde\alpha^{(1)}_{s^{(1)}}\sum_{s^{(0)}} \tilde{\omega}_{s^{(0)}\to s^{(1)}}^{(1)}x_{s^{(0)}}
   \\=&\frac{1}{C^{(0)}C^{(1)}}\text{ReLU}
   \left(\sum_{s^{(0)}}\widetilde W^{(1)}_{s^{(0)}\to s^{(1)}}\hat{x}_{s^{(0)}}\right)
   \\=&
   \frac{1}{C^{(0)}C^{(1)}}\widetilde h^{(1)}_{s^{(1)}}
   \end{split}
\end{equation}
Next, expand Eq.(\ref{eq_qA1}) with $A^{(2)}_{s^{(2)}}$ by substituting the iteration as given in Eq.(\ref{eq_At_itera}).
Notice that the second term in Eq.(\ref{eq_At_itera}) is fixed, and substitute the obtained $\widetilde\alpha^{(1)}_{s^{(1)}}$, we find out that $\widetilde\alpha^{(2)}_{s^{(2)}}$ must maximize the product
\begin{equation}
    \prod_{s^{(2)}}\exp\left(
       \frac{\beta}{C^{(0)}C^{(1)}} 
    A^{(2)}_{s^{(2)}}\alpha^{(2)}_{s^{(2)}}
    \sum_{s^{(1)}}\tilde{\omega}_{s^{(1)}\to s^{(2)}}^{(2)}
    \widetilde h_{s^{(1)}}^{(1)}
        \right)
\end{equation}
Therefore, we obtain $\widetilde\alpha^{(2)}_{s^{(2)}}$ as
\begin{equation}
    \widetilde\alpha^{(2)}_{s^{(2)}}=
    \begin{cases} 
  1, & \sum_{s^{(1)}}\tilde{\omega}_{s^{(1)}\to s^{(2)}}^{(2)}
    \widetilde h_{s^{(1)}}^{(1)}\geq0 \\ 
  0,      & \sum_{s^{(1)}}\tilde{\omega}_{s^{(1)}\to s^{(2)}}^{(2)}
    \widetilde h_{s^{(1)}}^{(1)}<0
\end{cases}
\label{eq_alpha2}
\end{equation}
yielding that
\begin{equation}
   \begin{split}
&\widetilde\alpha^{(2)}_{s^{(2)}}\sum_{s^{(1)}}\tilde{\omega}_{s^{(1)}\to s^{(2)}}^{(2)}
    \widetilde\alpha^{(1)}_{s^{(1)}}\sum_{s^{(0)}} \tilde{\omega}_{s^{(0)}\to s^{(1)}}^{(1)}x_{s^{(0)}}
   \\=&\frac{1}{C^{(0)}C^{(1)}C^{(2)}}\text{ReLU}
   \left(\sum_{s^{(1)}}\widetilde W^{(2)}_{s^{(1)}\to s^{(2)}}\widetilde h_{s^{(1)}}^{(1)}\right)
   \\=&
   \frac{1}{C^{(0)}C^{(1)}C^{(2)}}\widetilde h^{(2)}_{s^{(2)}}
   \end{split}
\end{equation}
Repeating this process iteratively, we prove that $\widetilde{\boldsymbol{\alpha}}$ is exactly same as the activated ($\widetilde\alpha^{(t)}_{s^{(t)}}=1$)/inactivated ($\widetilde\alpha^{(t)}_{s^{(t)}}=0$) state of the corresponding deep ReLU network. Recalling Eq.(\ref{eq_ztalpha}), we have
\begin{equation}
    \mathbf{\widetilde z}^{(t)}(\widetilde{\boldsymbol{\alpha}})
    =\left(\prod_{\tau=0}^tC^{(\tau)}\right)^{-1}\widetilde{\mathbf h}^{(t)}
\end{equation}
As $\lim_{\beta\to+\infty} \langle\mathbf{\widetilde z}^{(t)}\rangle =\mathbf{\widetilde z}^{(t)}(\widetilde{\boldsymbol{\alpha}})$, Eq.(\ref{eq_exacth}) is proved.
Recalling that $g_k$ is the $k$-th component of ${\widetilde z}^{(T)}$, we have $\left(\prod_{\tau=0}^tC^{(\tau)}\right)g_k(\widetilde{\boldsymbol{\alpha}})=f_k$. 
As $\lim_{\beta\to+\infty} \langle g_k\rangle =g_k(\widetilde{\boldsymbol{\alpha}})$, and $\lim_{\beta\to+\infty} \langle g_k\rangle^2= \lim_{\beta\to+\infty} \langle (g_k)^2\rangle =(g_k(\widetilde{\boldsymbol{\alpha}}))^2$,  Eq.(\ref{eq_exact_loss}) is proved.
\end{proof}

To validate Theorem \ref{theorem1}, we benchmark a $1\times6\times6\times1$ ReLU network (see Fig.(\ref{fig2}a) for the target function $\hat{y}=\sin(\pi \hat{x})+1$ and network prediction).
We then construct the corresponding physical ensemble and evaluate its observables across various temperatures $\beta=0, 1, 5, 10, 50, 100$.
In Fig.(\ref{fig2}b), the ensemble loss $\left\langle\left(f- \left(\prod_{\tau=0}^TC^{(\tau)}\right)g\right)^2\right\rangle_{\alpha}$ is compared with the exact ReLU training loss. 
Remarkably, for $\beta \geq 10$, the two losses coincide, validating Eq.~(\ref{eq_exact_loss}). 
Fig.(\ref{fig2}c) displays the rescaled output $\left\langle \left(\prod_{\tau=0}^TC^{(\tau)}\right)g \right\rangle_{\alpha}$ as a function of $\hat{x}$.
At $\beta=0$ (infinite temperature), all $\boldsymbol{\alpha}$ share same probability, yielding a trivial linear response. 
As $\beta$ increases, the system condenses towards the exact ReLU output.
Moreover, Fig.(\ref{fig1}d) presents the maximum probability $\max q(\boldsymbol{\alpha})$ as a function of the input $\hat{x}$. For large $\beta$ ($\beta \ge 50$), $\max q(\boldsymbol{\alpha})$ approaches unity across the majority of the input domain, confirming the complete condensation into the optimal $\boldsymbol{\widetilde \alpha}$.
We also notice that there are still a few sharp dips.
These dips are attributed to nearly degenerate suboptimal configurations $\boldsymbol{\alpha}$ that yield population imbalances closely matching the optimal $\boldsymbol{\widetilde \alpha}$.

Crucially, this condensation is accompanied by a first-order phase transition.
In Fig.(\ref{fig_phase}), we plot the derivative $\langle\partial\Pi/\partial x\rangle$, against $\hat{x}\in[-0.7, -0.1]$.
At $\beta=0$, the curve is flat. 
However, as $\beta$ increases, the curve develops sharp discontinuities. 
At $\beta=100$, the response exhibits step-like jumps.
By comparing with Fig.(\ref{fig2}a), these exhibits discontinuities precisely coincide with the transitions between different linear pieces of the ReLU network. 
Therefore, the piecewise linear behavior of deep ReLU networks originates from first-order phase transitions in the proposed fermionic transport model.

\begin{figure}[t]
    \centering
    \includegraphics[width=0.48\textwidth]{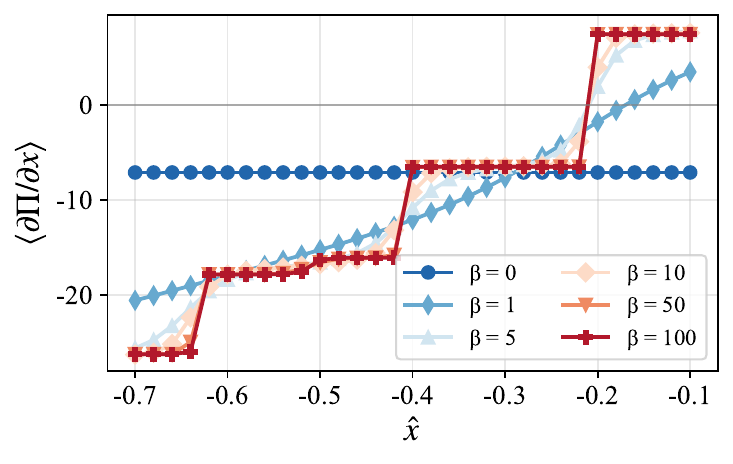}
    \caption{\textbf{First-order phase transition underlying ReLU piecewise linearity.} The derivative $\langle\partial\Pi/\partial x\rangle$ as a function of input $\hat{x}$ for various $\beta$.
    At low $\beta$, The derivative is continuous ($\beta=0$, flat line).
    As $\beta$ increases, discontinuities emerge ($\beta=100$, red crosses). These discontinuities align exactly with the boundaries of the linear pieces in the ReLU network (see Fig.~\ref{fig1}(a)), demonstrating that the piecewise linear behavior is a macroscopic manifestation of a first-order phase transition.}
    \label{fig_phase}
\end{figure}

\section{Conclusions}
We have established an exact correspondence between an underlying statistical-mechanical system and arbitrary deep ReLU networks.
Starting from a microscopic configuration space with well-defined state multiplicities, the model admits equivalent realizations as cascaded quantum operations with post-selection and as directed fermionic transport.
Without imposing a ReLU activation function, its equilibrium observables converge in the $\beta\to+\infty$ limit to the hidden states, output, and loss of the corresponding ReLU network, for arbitrary inputs, depths, and layer widths.

The correspondence further reveals a physical origin of the piecewise-linear structure of ReLU computation.
At $\beta\to+\infty$, the ensemble condenses into the activation configuration selected by the network, while changes between distinct activation patterns appear as discontinuous changes of the thermodynamic response.
These transition points coincide with the boundaries between linear regions of the ReLU network, establishing a direct correspondence between neural activation sectors and physical phases.
Therefore, we show that deep ReLU computation emerges as the zero-temperature behavior of physical systems.
This connection establishes a framework in which properties of neural architectures can be investigated through statistical-mechanical quantities and suggests a route toward designing computational architectures from microscopic physical principles.

{\it Acknowledgments}
J.L gratefully acknowledges funding by National Natural Science Foundation of China (NSFC) under Grant No.12305012.

\bibliography{ref}

@article{lecun2015deep,
  title={Deep learning},
  author={LeCun, Yann and Bengio, Yoshua and Hinton, Geoffrey},
  journal={nature},
  volume={521},
  number={7553},
  pages={436--444},
  year={2015},
  publisher={Nature Publishing Group UK London}
}

@article{schmidhuber2015deep,
  title={Deep learning in neural networks: An overview},
  author={Schmidhuber, J{\"u}rgen},
  journal={Neural networks},
  volume={61},
  pages={85--117},
  year={2015},
  publisher={Elsevier}
}

@book{goodfellow2016deep,
  title={Deep learning},
  author={Goodfellow, Ian and Bengio, Yoshua and Courville, Aaron and Bengio, Yoshua},
  volume={1},
  number={2},
  year={2016},
  publisher={MIT press Cambridge}
}

@article{saxe2021if,
  title={If deep learning is the answer, what is the question?},
  author={Saxe, Andrew and Nelli, Stephanie and Summerfield, Christopher},
  journal={Nature Reviews Neuroscience},
  volume={22},
  number={1},
  pages={55--67},
  year={2021},
  publisher={Nature Publishing Group UK London}
}

@article{sejnowski2020unreasonable,
  title={The unreasonable effectiveness of deep learning in artificial intelligence},
  author={Sejnowski, Terrence J},
  journal={Proceedings of the National Academy of Sciences},
  volume={117},
  number={48},
  pages={30033--30038},
  year={2020},
  publisher={National Academy of Sciences}
}

@inproceedings{nair2010rectified,
  title={Rectified linear units improve restricted boltzmann machines},
  author={Nair, Vinod and Hinton, Geoffrey E},
  booktitle={Proceedings of the 27th international conference on machine learning (ICML-10)},
  pages={807--814},
  year={2010}
}

@inproceedings{glorot2011deep,
  title={Deep sparse rectifier neural networks},
  author={Glorot, Xavier and Bordes, Antoine and Bengio, Yoshua},
  booktitle={Proceedings of the fourteenth international conference on artificial intelligence and statistics},
  pages={315--323},
  year={2011},
  organization={JMLR Workshop and Conference Proceedings}
}

@article{montufar2014number,
  title={On the number of linear regions of deep neural networks},
  author={Mont{\'u}far, Guido and Pascanu, Razvan and Cho, Kyunghyun and Bengio, Yoshua},
  journal={Advances in neural information processing systems},
  volume={27},
  year={2014}
}

@article{pascanu2013number,
  title={On the number of response regions of deep feed forward networks with piece-wise linear activations},
  author={Pascanu, Razvan and Montufar, Guido and Bengio, Yoshua},
  journal={arXiv preprint arXiv:1312.6098},
  year={2013}
}

@inproceedings{serra2018bounding,
  title={Bounding and counting linear regions of deep neural networks},
  author={Serra, Thiago and Tjandraatmadja, Christian and Ramalingam, Srikumar},
  booktitle={International conference on machine learning},
  pages={4558--4566},
  year={2018},
  organization={PMLR}
}

@inproceedings{laurent2018multilinear,
  title={The multilinear structure of ReLU networks},
  author={Laurent, Thomas and Brecht, James},
  booktitle={International conference on machine learning},
  pages={2908--2916},
  year={2018},
  organization={PMLR}
}

@article{he2020piecewise,
  title={Piecewise linear activations substantially shape the loss surfaces of neural networks},
  author={He, Fengxiang and Wang, Bohan and Tao, Dacheng},
  journal={arXiv preprint arXiv:2003.12236},
  year={2020}
}

@article{hopfield1982neural,
  title={Neural networks and physical systems with emergent collective computational abilities.},
  author={Hopfield, John J},
  journal={Proceedings of the national academy of sciences},
  volume={79},
  number={8},
  pages={2554--2558},
  year={1982}
}

@article{amit1985storing,
  title={Storing infinite numbers of patterns in a spin-glass model of neural networks},
  author={Amit, Daniel J and Gutfreund, Hanoch and Sompolinsky, Haim},
  journal={Physical review letters},
  volume={55},
  number={14},
  pages={1530},
  year={1985},
  publisher={APS}
}

@article{watkin1993statistical,
  title={The statistical mechanics of learning a rule},
  author={Watkin, Timothy LH and Rau, Albrecht and Biehl, Michael},
  journal={Reviews of Modern Physics},
  volume={65},
  number={2},
  pages={499},
  year={1993},
  publisher={APS}
}

@article{bahri2020statistical,
  title={Statistical mechanics of deep learning},
  author={Bahri, Yasaman and Kadmon, Jonathan and Pennington, Jeffrey and Schoenholz, Sam S and Sohl-Dickstein, Jascha and Ganguli, Surya},
  journal={Annual review of condensed matter physics},
  volume={11},
  number={1},
  pages={501--528},
  year={2020},
  publisher={Annual Reviews}
}

@article{li2021statistical,
  title={Statistical mechanics of deep linear neural networks: The backpropagating kernel renormalization},
  author={Li, Qianyi and Sompolinsky, Haim},
  journal={Physical Review X},
  volume={11},
  number={3},
  pages={031059},
  year={2021},
  publisher={APS}
}

@article{pacelli2023statistical,
  title={A statistical mechanics framework for Bayesian deep neural networks beyond the infinite-width limit},
  author={Pacelli, Rosalba and Ariosto, Sebastiano and Pastore, Mauro and Ginelli, Francesco and Gherardi, Marco and Rotondo, Pietro},
  journal={Nature Machine Intelligence},
  volume={5},
  number={12},
  pages={1497--1507},
  year={2023},
  publisher={Nature Publishing Group UK London}
}

@article{peterson1987mean,
  title={A mean field theory learning algorithm for neural networks},
  author={Peterson, Carsten and Anderson, James R},
  journal={Complex systems},
  volume={1},
  number={5},
  year={1987}
}

@inproceedings{mei2019mean,
  title={Mean-field theory of two-layers neural networks: dimension-free bounds and kernel limit},
  author={Mei, Song and Misiakiewicz, Theodor and Montanari, Andrea},
  booktitle={Conference on learning theory},
  pages={2388--2464},
  year={2019},
  organization={PMLR}
}

@article{sirignano2020mean,
  title={Mean field analysis of neural networks: A law of large numbers},
  author={Sirignano, Justin and Spiliopoulos, Konstantinos},
  journal={SIAM Journal on Applied Mathematics},
  volume={80},
  number={2},
  pages={725--752},
  year={2020},
  publisher={SIAM}
}

@article{oostwal2021hidden,
  title={Hidden unit specialization in layered neural networks: ReLU vs. sigmoidal activation},
  author={Oostwal, Elisa and Straat, Michiel and Biehl, Michael},
  journal={Physica A: Statistical Mechanics and its Applications},
  volume={564},
  pages={125517},
  year={2021},
  publisher={Elsevier}
}

@article{luo2021phase,
  title={Phase diagram for two-layer relu neural networks at infinite-width limit},
  author={Luo, Tao and Xu, Zhi-Qin John and Ma, Zheng and Zhang, Yaoyu},
  journal={Journal of Machine Learning Research},
  volume={22},
  number={71},
  pages={1--47},
  year={2021}
}

@inproceedings{baity2018comparing,
  title={Comparing dynamics: Deep neural networks versus glassy systems},
  author={Baity-Jesi, Marco and Sagun, Levent and Geiger, Mario and Spigler, Stefano and Arous, G{\'e}rard Ben and Cammarota, Chiara and LeCun, Yann and Wyart, Matthieu and Biroli, Giulio},
  booktitle={International Conference on Machine Learning},
  pages={314--323},
  year={2018},
  organization={PMLR}
}

@article{geiger2019jamming,
  title={Jamming transition as a paradigm to understand the loss landscape of deep neural networks},
  author={Geiger, Mario and Spigler, Stefano and d'Ascoli, St{\'e}phane and Sagun, Levent and Baity-Jesi, Marco and Biroli, Giulio and Wyart, Matthieu},
  journal={Physical Review E},
  volume={100},
  number={1},
  pages={012115},
  year={2019},
  publisher={APS}
}

@article{geiger2020scaling,
  title={Scaling description of generalization with number of parameters in deep learning},
  author={Geiger, Mario and Jacot, Arthur and Spigler, Stefano and Gabriel, Franck and Sagun, Levent and d’Ascoli, St{\'e}phane and Biroli, Giulio and Hongler, Cl{\'e}ment and Wyart, Matthieu},
  journal={Journal of Statistical Mechanics: Theory and Experiment},
  volume={2020},
  number={2},
  pages={023401},
  year={2020},
  publisher={IOP Publishing and SISSA}
}

@article{ghio2024sampling,
  title={Sampling with flows, diffusion, and autoregressive neural networks from a spin-glass perspective},
  author={Ghio, Davide and Dandi, Yatin and Krzakala, Florent and Zdeborov{\'a}, Lenka},
  journal={Proceedings of the National Academy of Sciences},
  volume={121},
  number={27},
  pages={e2311810121},
  year={2024},
  publisher={National Academy of Sciences}
}

@article{aaronson2005quantum,
  title={Quantum computing, postselection, and probabilistic polynomial-time},
  author={Aaronson, Scott},
  journal={Proceedings: Mathematics, Physical and Engineering},
  pages={3473--3482},
  year={2005},
  publisher={JSTOR}
}

@article{aliferis2007accuracy,
  title={Accuracy threshold for postselected quantum computation},
  author={Aliferis, Panos and Gottesman, Daniel and Preskill, John},
  journal={arXiv preprint quant-ph/0703264},
  year={2007}
}

@article{yoshioka2019transforming,
  title={Transforming generalized Ising models into Boltzmann machines},
  author={Yoshioka, Nobuyuki and Akagi, Yutaka and Katsura, Hosho},
  journal={Physical Review E},
  volume={99},
  number={3},
  pages={032113},
  year={2019},
  publisher={APS}
}

@article{wright2022deep,
  title={Deep physical neural networks trained with backpropagation},
  author={Wright, Logan G and Onodera, Tatsuhiro and Stein, Martin M and Wang, Tianyu and Schachter, Darren T and Hu, Zoey and McMahon, Peter L},
  journal={Nature},
  volume={601},
  number={7894},
  pages={549--555},
  year={2022},
  publisher={Nature Publishing Group UK London}
}

@article{abbasi2024physical,
  title={Physical activation functions (PAFs): An approach for more efficient induction of physics into physics-informed neural networks (PINNs)},
  author={Abbasi, Jassem and Andersen, P{\aa}l {\O}steb{\o}},
  journal={Neurocomputing},
  volume={608},
  pages={128352},
  year={2024},
  publisher={Elsevier}
}

@article{tao2022piecewise,
  title={Piecewise linear neural networks and deep learning},
  author={Tao, Qinghua and Li, Li and Huang, Xiaolin and Xi, Xiangming and Wang, Shuning and Suykens, Johan AK},
  journal={Nature Reviews Methods Primers},
  volume={2},
  number={1},
  pages={42},
  year={2022},
  publisher={Nature Publishing Group UK London}
}

@article{he2020relu,
  title={ReLU deep neural networks and linear finite elements},
  author={He, Juncai and Li, Lin and Xu, Jinchao and Zheng, Chunyue},
  journal={Journal of Computational Mathematics},
  volume={38},
  number={3},
  pages={502--527},
  year={2020},
  publisher={JSTOR}
}

\clearpage 
\onecolumngrid 
\begin{center}
    \textbf{\large Supplementary Materials}
\end{center}

\beginsupplement 


\section{Physical realization of affine transformations}

\subsection{Affine augmentation of a ReLU network}

Consider an arbitrary $T$-layer ReLU network with input $\mathbf{\hat{x}}$,
\begin{equation}
    \mathbf{h}^{(0)}=\mathbf{\hat{x}},
    \qquad
    \mathbf{h}^{(t)}
    =
    \operatorname{ReLU}
    \left(
        W^{(t)}\mathbf{h}^{(t-1)}
        +
        \mathbf{b}^{(t)}
    \right),
    \qquad
    t=1,\ldots,T,
\end{equation}
where
\begin{equation}
    W^{(t)}\in\mathbb{R}^{D_t\times D_{t-1}},
    \qquad
    \mathbf{b}^{(t)}\in\mathbb{R}^{D_t}.
\end{equation}
The superscript $t$ labels the network layer throughout this Supplemental
Material.

To represent the affine transformation by a linear map, we introduce the
augmented vectors
\begin{equation}
    \widetilde{\mathbf{h}}^{(t)}
    =
    \begin{pmatrix}
        \mathbf{h}^{(t)}\\
        1
    \end{pmatrix}
    \in\mathbb{R}^{D_t+1},
\end{equation}
and the corresponding affine matrices
\begin{equation}
    \widetilde{W}^{(t)}
    =
    \begin{pmatrix}
        W^{(t)} & \mathbf{b}^{(t)}\\
        \mathbf{0}^{\mathsf T} & 1
    \end{pmatrix}.
\end{equation}
Especially, the input $\mathbf{\hat{x}}$ is embedded into $\widetilde{\mathbf{h}}^{(0)}$,
\begin{equation}
    \widetilde{\mathbf{h}}^{(0)}
    =
    \begin{pmatrix}
        \mathbf{\hat{x}}\\
        1
    \end{pmatrix}
\end{equation}
Thus,
\begin{equation}
    \widetilde{W}^{(t)}
    \widetilde{\mathbf{h}}^{(t-1)}
    =
    \begin{pmatrix}
        W^{(t)}\mathbf{h}^{(t-1)}
        +
        \mathbf{b}^{(t)}
        \\
        1
    \end{pmatrix}.
\end{equation}
The last component is invariant under the transformation.
Consequently, the affine transformations of a ReLU layer is converted to linear transformations $\widetilde{W}^{(t)}$.
Hereafter, we will show how to map these linear transformations into quantum computing platforms, where we embed $\widetilde{W}^{(t)}$ into unitary operations, and the inputs are mapped as mixed states.
We denote $\mathbf{x}=\widetilde{\mathbf h}^{(0)}/C^{(0)}$ as the rescaled inputs, where $C^{(0)}\geq \sum_{j=1}^{d_0}|h^{(0)}_j|$.

\subsection{Embedding Linear transformation into two unitary operations}

Let
\begin{equation}
    M
    \in
    \mathbb{R}^{d_{\mathrm{out}}\times d_{\mathrm{in}}}
\end{equation}
be an arbitrary real matrix. We decompose it elementwise as
\begin{equation}
    M=M_{+}-M_{-},
\end{equation}
where
\begin{equation}
    (M_{+})_{kj}
    =
    \max(M_{kj},0),
    \qquad
    (M_{-})_{kj}
    =
    \max(-M_{kj},0),
\end{equation}
so that
\begin{equation}
    M_{+}\geq0,
    \qquad
    M_{-}\geq0.
\end{equation}

For the $t$-th network layer we apply this decomposition to the augmented
matrix,
\begin{equation}
    \widetilde{W}^{(t)}
    =
    \widetilde{W}^{(t)}_{+}
    -
    \widetilde{W}^{(t)}_{-}.
\end{equation}
We then introduce the rescaling factor
\begin{equation}
    C^{(t)}
    \geq
    \max
    \left\{
        \left\|
            \widetilde{W}^{(t)}_{+}
        \right\|_2,
        \left\|
            \widetilde{W}^{(t)}_{-}
        \right\|_2
    \right\},
\end{equation}
The L-2 norm is non negative, and is 0 iff all components of the matrix is 0.
To the end, we assume that $C^{(t)}>0$. 
The rescaled matrices
\begin{equation}
    \widetilde{\omega}^{(t)}_{\pm}
    =
    \frac{
        \widetilde{W}^{(t)}_{\pm}
    }{
        2C^{(t)}
    }
\end{equation}
therefore satisfy
\begin{equation}
    \left\|
        \widetilde{\omega}^{(t)}_{\pm}
    \right\|_2
    \leq1.
\end{equation}
Hence $\widetilde{\omega}^{(t)}_{\pm}$ are contractions. The affine matrix
can consequently be recovered as
\begin{equation}
    \widetilde{W}^{(t)}
    =
    \frac12 C^{(t)}
    \left(
        \widetilde{\omega}^{(t)}_{+}
        -
        \widetilde{\omega}^{(t)}_{-}
    \right).
    \label{eqs_affine_decomp}
\end{equation}

Consider  $\widetilde{\omega}\pm\in\mathbb{R}^{d_{\mathrm{out}}\times d_{\mathrm{in}}}$.
As $\|\widetilde{\omega}\pm\|_2\leq1$,
The contraction can be embedded into a unitary operator on a larger Hilbert space.
In particular, after enlarging the input and output spaces by ancillary degrees of freedom, there exists a unitary operator $\omega_\pm$,
whose nontrivial block coincides with $\widetilde{\omega}$,
\begin{equation}
    \widetilde{\omega}_\pm
    =
    P_{\mathrm{out}}\,
    \omega_\pm\,
    P_{\mathrm{in}},
    \label{eqs_block}
\end{equation}
where $P_{\mathrm{in}}$ and $P_{\mathrm{out}}$ project onto the computational
subspaces retained after post-selection.

For completeness, a standard unitary dilation can be written in block form
as
\begin{equation}
    \omega\pm
    =
    \begin{pmatrix}
        \widetilde{\omega}_\pm
        &
        \sqrt{
            I-\widetilde{\omega}_\pm\widetilde{\omega}_\pm^{\mathsf T}
        }
        \\[2mm]
        \sqrt{
            I-\widetilde{\omega}_\pm^{\mathsf T}\widetilde{\omega}_\pm
        }
        &
        -\widetilde{\omega}_\pm^{\mathsf T}
    \end{pmatrix},
    \label{eqs_dilation}
\end{equation}
with the block dimensions enlarged as necessary when $d_{\mathrm{in}}\neq d_{\mathrm{out}}$. 
The contraction condition guarantees
that the square-root operators are well defined, and the resulting matrix is
unitary on the enlarged space.

\subsection{Single-block signed transformation}

We next show how the controlled-unitary circuit converts the two unitary
branches into a signed linear transformation.

Consider a target register with computational basis
$\{|j\rangle\}_{j=1}^{d}$ and a control qubit $q_C$. The input state is
\begin{equation}
    \rho^{(\mathrm{in})}
    =
    \sum_{j=1}^{d}
    \left(
        z_{j,0}^{(in)}|0\rangle\langle0|_C
        +
        z_{j,1}^{(in)}|1\rangle\langle1|_C
    \right)
    \otimes
    |j\rangle\langle j|_T,
    \label{eqs_input}
\end{equation}
with
\begin{equation}
    z_{j,0}^{(in)},z_{j,1}^{(in)}\geq0,
    \qquad
    \sum_{j=1}^{d}
    (z_{j,0}^{(in)}+z_{j,1}^{(in)})=1.
\end{equation}
We define the real vector $\mathbf{z}^{(in)}$ by this ensemble as
\begin{equation}
    z^{(in)}_j
    =
    \operatorname{Tr}
\left[
\left(
Z_C\otimes |j\rangle\langle j|_T
\right)
\rho^{(\mathrm{in})}
\right]
    =
    z_{j,0}^{(in)}-z_{j,1}^{(in)},
\end{equation}
where $Z_C$ is the Pauli-Z operator (not the partition function, we only use $Z$ to represent Pauli-Z operation in this subsection), and the subscript indicates that it acts on the control qubit.
Equivalently, we have
\begin{equation}
    \mathbf{z}^{(in)}
    =
    \sum_{\sigma=0,1}
    (-1)^\sigma
    \left(
        z_{1,\sigma}^{(in)},
        \ldots,
        z_{d,\sigma}^{(in)}
    \right).
    \label{eqs_signed_vector}
\end{equation}
$\mathbf{z}^{(in)}$ represents the population imbalance carried by target state \(j\), with the control qubit providing the sign degree of freedom, \(|0\rangle_C\) carries positive weight and \(|1\rangle_C\) carries negative weight.

Then apply the operations as depicted in Fig.(1a).
A Hadamard gate is applied to the control register,
\begin{equation}
    |\sigma\rangle
    \xrightarrow{H}
    \frac{1}{\sqrt2}
    \left(
        |0\rangle
        +
        (-1)^\sigma |1\rangle
    \right).
\end{equation}
The system is then acted upon by a controlled unitary,
\begin{equation}
    U_C
    =
    |0\rangle\langle0|_C\otimes U_0
    +
    |1\rangle\langle1|_C\otimes U_1.
    \label{eqs_controlledU}
\end{equation}
Similarly, we define the output vector $\mathbf{z}^{(out)}$ by this ensemble as
\begin{equation}
    z^{(out)}_k
    =
    \operatorname{Tr}
\left[
\left(
Z_C\otimes |k\rangle\langle k|_T
\right)
\rho^{(\mathrm{out})}
\right]
\end{equation}

which gives
\begin{align}
    z_k^{(\mathrm{out})}
    &=
    \frac12
    \sum_j
    \left(
        z_{j,0}^{(in)}-z_{j,1}^{(in)}
    \right)
    \left[
        |\langle k|U_0|j\rangle|^2
        -
        |\langle k|U_1|j\rangle|^2
    \right]
    \nonumber\\
    &=
    \frac12
    \sum_j
    z_j^{(in)}
    \left[
        |\langle k|U_0|j\rangle|^2
        -
        |\langle k|U_1|j\rangle|^2
    \right].
    \label{eqs_output}
\end{align}

Thus, we obtain
\begin{equation}
    \mathbf{z}^{(\mathrm{out})}
    =
    \mathbf{\omega}
    \mathbf{z}^{(\mathrm{in})},
\end{equation}
with
\begin{equation}
    \mathbf{\omega}_{kj}
    =
    \frac12
    \left(
        |\langle k|U_0|j\rangle|^2-|\langle k|U_1|j\rangle|^2
    \right).
    \label{eqs_Weff}
\end{equation}

In this paper, we focus on the `effective' trajectories with $\widetilde\omega$, so that $k=1,\ldots,d_t,\qquad j=1,\ldots,d_{t-1}$, which correspond to the affine transformation of the original ReLU network.
The other degrees of freedom then provide escape channels for the probability weight that is not retained by the effective transformation.

\subsection{Cascading the quantum blocks}
Now consider a cascading of $T$ blocks.
For simplicity, we denote the number of shots after applying the $t$-th block, with readout $\sigma^{(t)}=0,1$, $s^{(t)}=k$ as $N^{(t)}_{0,k}$ and $N^{(t)}_{1,k}$, respectively, and denote the total number of shots as $N$, initially.
We denote the population imbalance ${\mathbf z}^{(t)}\in\mathbb{R}^{d_t}$, with the $k$-th component as
\begin{equation}
    z^{(t)}_k=\frac{N^{(t)}_{0,k}-N^{(t)}_{1,k}}{N}, \qquad k=1,\ldots,d_t
\end{equation}
$z^{(t)}_k$ describes the difference of shots with same target qubit readout $k$, but with different control qubit readouts.
The input of the $t$-th block is given as
\begin{equation}
    \rho^{(t-1)}=
    \sum_{k=1}^{d_{t-1}}
    \left(\frac{N^{(t)}_{0,k}}{N}|0\rangle\langle0|_C\otimes|k\rangle\langle k|_T+
    \frac{N^{(t)}_{1,k}}{N}|1\rangle\langle1|_C\otimes|k\rangle\langle k|\right)
    +\frac{N-\sum_{k=1}^{d_{t-1}}
    \left(N^{(t)}_{0,k}+N^{(t)}_{1,k}\right)}{N}\rho^{(t-1)}_{others}
\end{equation}
where $\rho^{(t-1)}_{others}$ represents the density matrix of the discard shots in the first $t-1$ blocks.
The discard shots are not sent to the $t$-th block, so that the operations only acts on the first term.
Specifically, we denote $\mathbf{x}=\widetilde{\mathbf h}^{(0)}/C^{(0)}$ as the rescaled inputs, where $C^{(0)}\geq \sum_{j=1}^{d_0}|h^{(0)}_j|$.
$\mathbf{x}$ is mapped as a mixed state
\begin{equation}
    \rho^{(0)}
    =\sum_{k=1}^{d_0}\left(
    |x_k|\cdot|\Theta(x_k)\rangle\langle\Theta(x_k)|_C\otimes|k\rangle\langle k|_T
    \right)
    +\left((1-\sum_{k=1}^{d_0}|x_k|\right)\rho^{(0)}_{others}
\end{equation}
where $\Theta$ is the threshold function, $\Theta(x_k)=1$ if $x_k\geq0$, otherwise $\Theta(x_k)=0$ if $x_k<0$.

The $t$-th block therefore realizes linear transformation
\begin{equation}
{\mathbf z}^{(t)}
=
\widetilde \omega^{(t)}
{\mathbf z}^{(t-1)}
\end{equation}
After the controlled-unitary operation, all qubits are measured in the computational basis. 
The resulting conditional state is still diagonal in the computational basis.
Record the readout of control qubit as $\sigma^{(t)}=0,1$, and target qubit as $s^{(t)}=1,\ldots, d_{t},d_{t}+1,\ldots$.
The shots with $s^{(t)}=1,\ldots, d_{t}$ are sent to the next block as input, as they correspond to the transformation of $\widetilde W^{(t)}$.
The other shots are discard, and will no more be sent to the successive blocks.
Thus the retained state therefore has the same diagonal structure as Eq.~\eqref{eqs_input} and can serve as the input to another block.
The cascade of $T$ blocks implements
\begin{equation}
    \begin{split}
    {\mathbf z}^{(t)}
    &=\widetilde \omega^{(T)}
    \widetilde \omega^{(T-1)}
    \cdots
    \widetilde \omega^{(1)}
    {\mathbf z}^{(0)}
    \\&=
    \left(
        \prod_{t=1}^{T}C^{(t)}
    \right)^{-1}
    \widetilde W^{(T)}
    \widetilde W^{(T-1)}
    \cdots
    \widetilde W^{(1)}
    {\mathbf z}^{(0)}.
    \end{split}
    \label{eqs_cascade}
\end{equation}

\section{Fermionic Transport Model Interpretation}

\subsection{Effective trajectories and probability}

For a shot that completes the $T$ cascaded blocks, record the measurement results of control qubit
\begin{equation}
    \boldsymbol{\sigma}
    =
    \left(
        \sigma^{(0)},
        \sigma^{(1)},
        \ldots,
        \sigma^{(T)}
    \right),
    \qquad
    \sigma^{(t)}\in\{0,1\},
\end{equation}
and record the results of the target qubits
\begin{equation}
    \mathbf{s}
    =
    \left(
        s^{(0)},
        s^{(1)},
        \ldots,
        s^{(T)}
    \right),
    \qquad
    s^{(t)}\in\{1,\ldots,d_t\}.
\end{equation}
$\boldsymbol{\sigma},\mathbf{s}$ together define a complete trajectory.
Such trajectories are `effective', as they correspond to the affine transformations of the original ReLU network.
Shots with $s^{(t)}>d_t$ are discard, and will not be sent to the successive blocks as inputs.

Denote the probability of an effective trajectory by
\begin{equation}
    p(\boldsymbol{\sigma},\mathbf{s})
    =
    \Pr
    \left(
        \sigma^{(0)}\to\cdots\to\sigma^{(T)},
        \;
        s^{(0)}\to\cdots\to s^{(T)}
    \right).
\end{equation}
The signed output associated with the final target state $k$ is then
\begin{equation}
    z_k^{(T)}
    =
    \sum_{\boldsymbol{\sigma},\mathbf{s}}
    (-1)^{
        \sum_{t=0}^{T}\sigma^{(t)}
    }
    \delta_{s^{(T)},k}
    p(\boldsymbol{\sigma},\mathbf{s}).
    \label{eqs_trajectory_output}
\end{equation}
The probability of the full trajectory can be split into each steps,
\begin{equation}
    \begin{split}
    p(\boldsymbol{\sigma},\mathbf{s})
    =
    &\Pr\left(
    \sigma^{(0)} 
    \to
    \sigma^{(1)}
    \to
    \cdots
    \to
    \sigma^{(T)},
    \to
    s^{(1)}
    \to
    \cdots
    \to
    s^{(T)}
    \right)
    \\
    =&
    \Pr\left(
    \sigma^{(0)}
    \to
    \sigma^{(1)},
    s^{(0)}
    \to
    s^{(1)}
    \right)\cdots
    \Pr\left(
    \sigma^{(T-1)}
    \to
    \sigma^{(T)},
    s^{(T-1)}
    \to
    s^{(T)}
    \right)
    \\=&
    \prod_{\tau=1}^T
    \frac12 |\langle s^{(\tau)}|U^{(t)}_{\sigma^{(\tau)}}|s^{(\tau-1)}\rangle|^2
    \end{split}
\end{equation}

These trajectories can be interpreted as a directed layered fermionic transport network. 
Consider a two-dimensional network with $T$ columns, where column $t$ contains $d_t$ sites. 
The site index is identified with the target state, $s^{(t)}\in\{1,\ldots,d_t\}$.
The columns are ordered from left to right corresponding to the successive quantum blocks.

Each measurement shot corresponds to a fermion initially occupying a certain site in the first column. 
Its site is determined by the target-register state
and its spin by the control state,
\begin{equation}
    \sigma^{(t)}=0,1
    \quad\longleftrightarrow\quad
    \uparrow,\downarrow.
\end{equation}
At each layer, the fermion propagates from column $t-1$ to column $t$
The hopping probability from site $s^{(t-1)}$ in the $t-1$-th column, to site  $s^{(t)}$ in the $t$-th column, with spin $\sigma^{(t)}$, is generated by the corresponding quantum block as
\begin{equation}
    \Pr\left(
    \sigma^{(t-1)}
    \to
    \sigma^{(t)},
    s^{(t)}
    \to
    s^{(t)}
    \right)=
    \frac12 |\langle s^{(t)}|U_{\sigma^{(t)}}|s^{(t-1)}\rangle|^2
\end{equation}
When a fermion hops to the neighbor column, it has same probability to keep or flip its spin, yielding the $\frac12$.
However, depending on whether its spin is flipped, the probability of hopping to various sites are often different.
Therefore, it is expected to observe imbalanced population at each site after the hopping.
Recalling that there are only $d_t$ sites in the $t$-th column corresponding to $s^{(t)}\in\{1,\ldots,d_t\}$.
Shots with $s^{(t)}>d_t$ are interpreted as fermions that escape the 2-D network.

\subsection{Partition function for effective trajectories}
Focusing on the effective trajectories with rescaled input $\mathbf{x}$, the probability to find a shot/fermion with trajectory $\boldsymbol{\sigma},\mathbf{s}$ starting from $s^{(0)}$, and post selection criteria (or open/closed status) $\boldsymbol{\alpha}$ is $x_{s^{(0)}}p(\boldsymbol{\sigma},\mathbf{s})q(\boldsymbol{\alpha})$, as these two events are independent.
Therefore, the total entropy is given as
\begin{equation}
    \begin{split}
        S_{tot}=&-\sum_{\boldsymbol{\sigma},\mathbf{s}}\sum_{\boldsymbol{\alpha}}x_{s^{(0)}}p(\boldsymbol{\sigma},\mathbf{s})q(\boldsymbol{\alpha})\ln\left(x_{s^{(0)}}p(\boldsymbol{\sigma},\mathbf{s})q(\boldsymbol{\alpha})\right)
        \\=&-\sum_{\boldsymbol{\sigma},\mathbf{s}}x_{s^{(0)}}p(\boldsymbol{\sigma},\mathbf{s})
        \sum_{\boldsymbol{\alpha}}q(\boldsymbol{\alpha})\ln q(\boldsymbol{\alpha})
        -\sum_{\boldsymbol{\alpha}}q(\boldsymbol{\alpha})\sum_{\boldsymbol{\sigma},\mathbf{s}}x_{s^{(0)}}p(\boldsymbol{\sigma},\mathbf{s})\ln\left(x_{s^{(0)}}p(\boldsymbol{\sigma},\mathbf{s})\right)
        \\=&-\sum_{\boldsymbol{\sigma},\mathbf{s}}x_{s^{(0)}}p(\boldsymbol{\sigma},\mathbf{s})
        \sum_{\boldsymbol{\alpha}}q(\boldsymbol{\alpha})\ln q(\boldsymbol{\alpha})-\sum_{\boldsymbol{\sigma},\mathbf{s}}x_{s^{(0)}}p(\boldsymbol{\sigma},\mathbf{s})\ln\left(x_{s^{(0)}}p(\boldsymbol{\sigma},\mathbf{s})\right)
    \end{split}
\end{equation}
for fixed quantum blocks and inputs, $\sum_{\boldsymbol{\sigma},\mathbf{s}}x_{s^{(0)}}p(\boldsymbol{\sigma},\mathbf{s})\ln\left(x_{s^{(0)}}p(\boldsymbol{\sigma},\mathbf{s})\right)$ are constants, yielding that the second term is fixed, with $\sum_{\boldsymbol{\sigma},\mathbf{s}}x_{s^{(0)}}p(\boldsymbol{\sigma},\mathbf{s})$ as a rescaling constant.
For simplicity we focus on the Shannon entropy of $\boldsymbol{\alpha}$
\begin{equation}
    S = -\sum_{\boldsymbol{\alpha}} q(\boldsymbol{\alpha}) \ln q(\boldsymbol{\alpha}).
\end{equation}
where the rescaling constant can be combined into $\beta$ in the final results.

Next we need to find a constraint.
Recalling that the population imbalance
\begin{equation}
    \mathbf{\widetilde z}^{(t)}(\boldsymbol{\alpha})
    = \hat{\alpha}^{(t)}\widetilde \omega^{(t)}
    \hat{\alpha}^{(t-1)}\widetilde \omega^{(t-1)}
    \cdots
    \hat{\alpha}^{(1)}\widetilde \omega^{(1)}
    {\mathbf z}^{(0)}
    \label{eqs_ztalpha}
\end{equation}
where $\hat{\alpha}^{(t)}=diag\{\alpha^{(t)}_1, \ldots, \alpha^{(t)}_{d_t}\}$.
Especially, the $k$-th component of $\mathbf{\widetilde z}^{(T)}(\boldsymbol{\alpha})$ is given by
\begin{equation}
    \widetilde z_k^{(T)}(\boldsymbol{\alpha})
    =
    \sum_{\boldsymbol{\sigma},\mathbf{s}}
    (-1)^{\sum_{t=0}^{T}\sigma^{(t)}}
    \delta_{s_T,k}\,
    p(\boldsymbol{\sigma},\mathbf{s})
    \prod_{t=1}^{T}
    \alpha^{(t)}_{s^{(t)}}.
\end{equation}
In the fermionic transport picture, $\widetilde{z}_{k}^{(t)}$ denotes the ensemble-averaged signed spin associated with the $k$-th site in the $t$-th column.
Each fermion contributes $+1/2$ if it passes through this site with spin up, $-1/2$ if it passes through it with spin down, and zero otherwise.
The average is taken over the entire ensemble, including fermions that fulfill an effective trajectory, and those that escape from the network.
Thus, the overall population imbalance $\mathbf{\widetilde z}^{(t)}$ contains trivial shots, with possible $s^{(\tau)}>d_\tau$ in the successive $T-t$ blocks $\tau=t+1,\ldots,T$.

However, we only need the effective shots, $\sigma^{(t)}\in\{0,1\},s^{(t)}\in\{1,\ldots,d_t\}$, which correspond to affine transformations of the original ReLU network, and fermions that always stay in the network.
The probability that a fermion keeps in the 2-D network (not escape) starting from the $k$-th site in the $t$-th column is (fermions that flees away the network are not effective, yet the ones blocked at certain sites due to $\alpha=0$ are still effective, and their contributions should be included)
\begin{equation}
    \sum_{s^{(t+1)}}\cdots\sum_{s^{(T)}}\Omega_{s^{(t)}=k\to s^{(t+1)}}^{(t+1)}\Omega_{s^{(t+1)}\to s^{(t+2)}}^{(t+2)}\cdots\Omega_{s^{(T-1)}\to s^{(T)}}^{(T)}
\end{equation}
where $\Omega^{(\tau)}=\widetilde\omega_+ + \widetilde\omega_-$, and $\Omega_{s^{(\tau-1)}\to s^{(\tau)}}^{(\tau)}=\Omega_{s^{(\tau)},s^{(\tau-1)}}^{(\tau)}$, we rewrite the subscripts to express the propagation direction of the network.
Therefore, we have the effective population imbalance ${\pi}^{(t)}_k$ for the $k$-th site in the $t$-th column,
\begin{equation}
    \begin{split}
        {\pi}^{(t)}_k&={\widetilde z}^{(t)}_k \sum_{s^{(t+1)}}\cdots\sum_{s^{(T)}}\Omega_{s^{(t)}=k\to s^{(t+1)}}^{(t+1)}\Omega_{s^{(t+1)}\to s^{(t+2)}}^{(t+2)}\cdots\Omega_{s^{(T-1)}\to s^{(T)}}^{(T)}
        \\&=
        \sum_{\boldsymbol{s},s^{(t)}=k}
        \left(
        \prod_{\tau=1}^t
        \tilde{\omega}_{s^{(\tau-1)}\to s^{(\tau)}}^{(\tau)}
        \right)
        \left(
        \prod_{\tau=t+1}^T
        {\Omega}_{s^{(\tau-1)}\to s^{(\tau)}}^{(\tau)}
        \right)
        x_{s^{(0)}}\prod_{\tau=1}^{t}
    \alpha^{(\tau)}_{s^{(\tau)}}
    \\&=
        \sum_{\boldsymbol{\mathbf{s}} ,s^{(t)}=k}\left\{
        \prod_{\tau=1}^t\frac12
        \left(
        |\langle s^{(\tau)}|U^{(t)}_{0}|s^{(\tau-1)}\rangle|^2
        -|\langle s^{(\tau)}|U^{(t)}_{1}|s^{(\tau-1)}\rangle|^2
        \right)
        \right.
        \\&
        \left.
        \cdot\prod_{\tau=t+1}^T\frac{1}{2}
        \left(
        |\langle s^{(\tau)}|U^{(t)}_{0}|s^{(\tau-1)}\rangle|^2
        +|\langle s^{(\tau)}|U^{(t)}_{1}|s^{(\tau-1)}\rangle|^2
        \right)
        x_{s^{(0)}}\prod_{\tau=1}^{t}
    \alpha^{(\tau)}_{s^{(\tau)}}\right\}
    \\&=
    \sum_{\boldsymbol{\sigma}, \mathbf{s},s^{(t)}=k} 
        \left(
        \prod_{\tau=1}^T\frac12 |\langle s^{(\tau)}|U^{(t)}_{\sigma^{(\tau)}}|s^{(\tau-1)}\rangle|^2
        \right)
        x_{s_0} \prod_{\tau=1}^t (-1)^{\sigma^{(\tau)}}
    \alpha^{(\tau)}_{s^{(\tau)}}
    \\&=
    \sum_{\boldsymbol{\sigma},\mathbf{s},s^{(t)}=k}
    p(\boldsymbol{\sigma},\mathbf{s})x_{s_0}
    \prod_{\tau=1}^{t}(-1)^{\sigma^{(\tau)}}
    \alpha^{(\tau)}_{s^{(\tau)}}
    \end{split}
\end{equation}
For simplicity, we define the effective population imbalance ${\pi}^{(t)}$ of all effective sites of the $t$-th layer,
\begin{equation}
    \pi^{(t)}(\boldsymbol{\alpha}) = \sum_{k=1}^{d_t}{\pi}^{(t)}_k=
    \sum_{\boldsymbol{\sigma},\mathbf{s}}
    p(\boldsymbol{\sigma},\mathbf{s})x_{s_0}
    \prod_{\tau=1}^{t}(-1)^{\sigma^{(\tau)}}
    \alpha^{(\tau)}_{s^{(\tau)}}
\end{equation}
where the sum runs over all effective trajectories that pass by $s^{(t)}=k$.
In fermionic transport model, $\pi^{(t)}$ describes the average spin of sites $s^{(t)}=1,\ldots,d_t$ in $t$-th layer over all effective trajectories.

\section{Supplementary materials of the proposed Theorem and the proof}
\subsection{Preliminaries}
Recalling that the effective population imbalance of all effective sites of the $t$-th layer, ${\pi}^{(t)}$, can be rewritten as
\begin{equation}
    \begin{split}
    \pi^{(t)}(\boldsymbol{\alpha}) &= 
    \sum_{\boldsymbol{\sigma},\mathbf{s}}
    p(\boldsymbol{\sigma},\mathbf{s})x_{s_0}
    \prod_{\tau=1}^{t}(-1)^{\sigma^{(\tau)}}
    \alpha^{(\tau)}_{s^{(\tau)}}
        \\
        =&
        \sum_{\boldsymbol{\sigma}, \mathbf{s}} 
        \left(
        \prod_{\tau=1}^T\frac12 |\langle s^{(\tau)}|U^{(t)}_{\sigma^{(\tau)}}|s^{(\tau-1)}\rangle|^2
        \right)
        x_{s_0} \prod_{\tau=1}^t (-1)^{\sigma^{(\tau)}}
    \alpha^{(\tau)}_{s^{(\tau)}}
        \\
        =&
        \sum_{\boldsymbol{\mathbf{s}} }\left\{
        \prod_{\tau=1}^t\frac12
        \left(
        |\langle s^{(\tau)}|U^{(t)}_{0}|s^{(\tau-1)}\rangle|^2
        -|\langle s^{(\tau)}|U^{(t)}_{1}|s^{(\tau-1)}\rangle|^2
        \right)
        \right.
        \\&
        \left.
        \cdot\prod_{\tau=t+1}^T\frac{1}{2}
        \left(
        |\langle s^{(\tau)}|U^{(t)}_{0}|s^{(\tau-1)}\rangle|^2
        +|\langle s^{(\tau)}|U^{(t)}_{1}|s^{(\tau-1)}\rangle|^2
        \right)
        x_{s^{(0)}}\prod_{\tau=1}^{t}
    \alpha^{(\tau)}_{s^{(\tau)}}\right\}
        \\
        =&
        \sum_{s^{(0)}}\sum_{s^{(1)}}\cdots \sum_{s^{(T)}}
        \left(
        \prod_{\tau=1}^t
        \tilde{\omega}_{s^{(\tau-1)}\to s^{(\tau)}}^{(\tau)}
        \right)
        \left(
        \prod_{\tau=t+1}^T
        {\Omega}_{s^{(\tau-1)}\to s^{(\tau)}}^{(\tau)}
        \right)
        x_{s^{(0)}}\prod_{\tau=1}^{t}
    \alpha^{(\tau)}_{s^{(\tau)}}
    \end{split}
    \label{eqs_repi}
\end{equation}
where $\Omega^{(\tau)}=\widetilde\omega_+ + \widetilde\omega_-$, $\widetilde \omega^{(t)}=\left(\widetilde \omega^{(t)}_{+} - \widetilde \omega^{(t)}_{-}\right)/2$, with $\left(\widetilde\omega^{(t)}_{\pm}\right)_{kj}=|\langle k|U_{0,1}|j\rangle|^2$. 
We rewrite the subscripts to express the propagation direction of the ReLU network/hopping direction of the fermions, yielding that $\Omega_{s^{(\tau-1)}\to s^{(\tau)}}^{(\tau)}=\Omega_{s^{(\tau)},s^{(\tau-1)}}^{(\tau)}$, $\widetilde\omega_{s^{(\tau-1)}\to s^{(\tau)}}^{(\tau)}=\widetilde\omega_{s^{(\tau)},s^{(\tau-1)}}^{(\tau)}$.

Applying Eq.(\ref{eqs_repi}) we can expand $q(\boldsymbol{\alpha})$ as
\begin{equation}
    \begin{split}
    q(\boldsymbol{\alpha})=&\frac{1}{Z}\exp \left(\beta\sum_{t=1}^T  2^{T-t}\pi^{(t)}(\boldsymbol{\alpha})\right)
    \\=&\frac{1}{Z}\exp \left( \sum_{t=1}^T \beta 2^{T-t} \sum_{\boldsymbol{\sigma}, \mathbf{s}} p(\boldsymbol{\sigma}, \mathbf{s}) x_{s_0} \prod_{\tau=1}^t (-1)^{\sigma^{(\tau)}} \alpha^{(\tau)}_{s^{(\tau)}} \right)
    \\
    =&
    \frac{1}{Z}\exp \left\{ \beta\sum_{t=1}^T  2^{T-t} 
    \sum_{s^{(0)}}\sum_{s^{(1)}}\cdots \sum_{s^{(T)}}
        \left(
        \prod_{\tau=1}^t
        \tilde{\omega}_{s^{(\tau-1)}\to s^{(\tau)}}^{(\tau)}
        \right)
        \left(
        \prod_{\tau=t+1}^T
        {\Omega}_{s^{(\tau-1)}\to s^{(\tau)}}^{(\tau)}
        \right)
        x_{s^{(0)}}
    \prod_{\tau=1}^t \alpha^{(\tau)}_{s^{(\tau)}} \right\}
    \\
    =&
    \frac{1}{Z}\exp \left\{ \beta
    \sum_{s^{(0)}}\sum_{s^{(1)}}
    2^{T-1}\alpha^{(1)}_{s^{(1)}}
        \tilde{\omega}_{s^{(0)}\to s^{(1)}}^{(1)}
        x_{s^{(0)}}
        \left(\sum_{s^{(2)}}\cdots \sum_{s^{(T)}}
        \prod_{\tau=2}^T
        {\Omega}_{s^{(\tau-1)}\to s^{(\tau)}}^{(\tau)}
        \right)
        \right.
        \\
        &+
        \beta
    \sum_{s^{(0)}}\sum_{s^{(1)}}\sum_{s^{(2)}}
    2^{T-2}\alpha^{(1)}_{s^{(1)}}\alpha^{(2)}_{s^{(2)}}
        \tilde{\omega}_{s^{(0)}\to s^{(1)}}^{(1)}
        \tilde{\omega}_{s^{(1)}\to s^{(2)}}^{(2)}
        x_{s^{(0)}}
        \left(\sum_{s^{(3)}}\cdots \sum_{s^{(T)}}
        \prod_{\tau=3}^T
        {\Omega}_{s^{(\tau-1)}\to s^{(\tau)}}^{(\tau)}
        \right)
        \\
        &+\cdots \cdots
        \\
        &+
        \beta
    \sum_{s^{(0)}}\cdots\sum_{s^{(T-1)}}
    2\alpha^{(1)}_{s^{(1)}}
        \tilde{\omega}_{s^{(0)}\to s^{(1)}}^{(1)}
\cdots\alpha^{(T-1)}_{s^{(T-1)}}
        \tilde{\omega}_{s^{(T-2)}\to s^{(T-1)}}^{(T-1)}
        x_{s^{(0)}}
        \left(\sum_{s^{(T)}}
        {\Omega}_{s^{(T-1)}\to s^{(T)}}^{(T)}
        \right)
        \\
        &+\left.
        \beta
    \sum_{s^{(0)}}\cdots\sum_{s^{(T)}}
    \alpha^{(1)}_{s^{(1)}}
        \tilde{\omega}_{s^{(0)}\to s^{(1)}}^{(1)}
\cdots\alpha^{(T)}_{s^{(T)}}
        \tilde{\omega}_{s^{(T-2)}\to s^{(T)}}^{(T)}
        x_{s^{(0)}}
    \right\}
    \\
    =&
    \frac{1}{Z}\exp \left\{ \beta
    \sum_{s^{(0)}}\sum_{s^{(1)}}
    \alpha^{(1)}_{s^{(1)}}
        \tilde{\omega}_{s^{(0)}\to s^{(1)}}^{(1)}
        x_{s^{(0)}}
        \left[
        2^{T-1}\left(\sum_{s^{(2)}}\cdots \sum_{s^{(T)}}
        \prod_{\tau=2}^T
        {\Omega}_{s^{(\tau-1)}\to s^{(\tau)}}^{(\tau)}
        \right)
        \right.
        \right.
        \\
        &+
    \sum_{s^{(2)}}\alpha^{(2)}_{s^{(2)}}
        \tilde{\omega}_{s^{(1)}\to s^{(2)}}^{(2)}
        \left(
        2^{T-2}\left(\sum_{s^{(3)}}\cdots \sum_{s^{(T)}}
        \prod_{\tau=3}^T
        {\Omega}_{s^{(\tau-1)}\to s^{(\tau)}}^{(\tau)}
        \right)
        +\cdots
        \right.
        \\
        &+
        \left.\left.\left.
        \sum_{s^{(T-1)}}
    \alpha^{(T-1)}_{s^{(T-1)}}
        \tilde{\omega}_{s^{(T-2)}\to s^{(T-1)}}^{(T-1)}
        \sum_{s^{(T)}}\left(
        2{\Omega}_{s^{(T-1)}\to s^{(T)}}^{(T)}
        +
        \alpha^{(T)}_{s^{(T)}}
        \tilde{\omega}_{s^{(T-1)}\to s^{(T)}}^{(T)}
        \right)
        \right)\right]\right\}
    \end{split}
\end{equation}
 
For simplicity, here we denote $A^{(T)}_{s^{T}}=1$,
and iteratively, for $t=2,3,\ldots,T$, we have
\begin{equation}
    \begin{split}
        A^{(t-1)}_{s^{(t-1)}}=
        &\sum_{s^{(t)}}\alpha^{(t-1)}_{s^{(t-1)}}
\tilde{\omega}_{s^{(t-2)}\to s^{(t-1)}}^{(t-1)}A^{(t)}_{s^{(t)}}
+
2^{T-t+1}\sum_{s^{(t)}}\cdots\sum_{s^{(T)}}\prod_{\tau=t}^T{\Omega}_{s^{(\tau-1)}\to s^{(\tau)}}^{(\tau)}
    \end{split}
\label{eqs_At_itera}
\end{equation}

\begin{lemma}
For $t=1,2,\ldots,T-1$, $s^{(t)}=1,2\ldots,d_t$, we have 
\begin{equation}
    \left(2^{T-t+1}-1\right)\sum_{s^{(t+1)}}\cdots\sum_{s^{(T)}}\prod_{\tau=t+1}^T{\Omega}_{s^{(\tau-1)}\to s^{(\tau)}}^{(\tau)}
    \geq A^{(t)}_{s^{(t)}}\geq
    \sum_{s^{(t+1)}}\cdots\sum_{s^{(T)}}\prod_{\tau=t+1}^T{\Omega}_{s^{(\tau-1)}\to s^{(\tau)}}^{(\tau)}
\end{equation}
\label{lm_A}
\end{lemma}
\begin{proof}
As $A^{(T)}_{s^{T}}=1$, by Eq.(\ref{eq_At_itera}), we have
\begin{equation}
    A^{(T-1)}_{s^{T-1}} = \sum_{s^{(T)}}\left(
        2{\Omega}_{s^{(T-1)}\to s^{(T)}}^{(T)}
        +
        \alpha^{(T)}_{s^{(T)}}
        \tilde{\omega}_{s^{(T-1)}\to s^{(T)}}^{(T)}
        \right)
\end{equation}
Recalling that ${\Omega}_{s^{(\tau-1)}\to s^{(\tau)}}^{(\tau)}=\tilde{\omega}_{\sigma^{(\tau)}=0; s^{(\tau-1)}\to s^{(\tau)}}^{(\tau)}
        +\tilde{\omega}_{\sigma^{(\tau)}=1; s^{(\tau-1)}\to s^{(\tau)}}^{(\tau)}$, and
$\tilde{\omega}_{s^{(\tau-1)}\to s^{(\tau)}}^{(\tau)}=\tilde{\omega}_{\sigma^{(\tau)}=0; s^{(\tau-1)}\to s^{(\tau)}}^{(\tau)}
        -\tilde{\omega}_{\sigma^{(\tau)}=1; s^{(\tau-1)}\to s^{(\tau)}}^{(\tau)}$
As $\tilde{\omega}\geq0$, it is guaranteed that ${\Omega}_{s^{(\tau-1)}\to s^{(\tau)}}^{(\tau)}\geq |\tilde{\omega}_{s^{(\tau-1)}\to s^{(\tau)}}^{(\tau)}| \geq0$.
As $\alpha\in\{0,1\}$, we have
\begin{equation}
    3{\Omega}_{s^{(T-1)}\to s^{(T)}}^{(T)}\geq2{\Omega}_{s^{(T-1)}\to s^{(T)}}^{(T)}
        +
        \alpha^{(T)}_{s^{(T)}}
        \tilde{\omega}_{s^{(T-1)}\to s^{(T)}}^{(T)}\geq{\Omega}_{s^{(T-1)}\to s^{(T)}}^{(T)}
\end{equation}
Therefore, the bound holds for $t=T-1$.
\begin{equation}
    \left(2^2-1\right)\sum_{s^{(T)}}{\Omega}_{s^{(\tau-1)}\to s^{(\tau)}}^{(\tau)}
    \geq A^{(T-1)}_{s^{T-1}}\geq
    \sum_{s^{(T)}}{\Omega}_{s^{(\tau-1)}\to s^{(\tau)}}^{(\tau)}
\end{equation}

Assume the bound holds for $t$. 
\begin{equation}
    \left(2^{T-t+1}-1\right)\sum_{s^{(t+1)}}\cdots\sum_{s^{(T)}}\prod_{\tau=t+1}^T{\Omega}_{s^{(\tau-1)}\to s^{(\tau)}}^{(\tau)}
    \geq A^{(t)}_{s^{(t)}}\geq
    \sum_{s^{(t+1)}}\cdots\sum_{s^{(T)}}\prod_{\tau=t+1}^T{\Omega}_{s^{(\tau-1)}\to s^{(\tau)}}^{(\tau)}
\end{equation}
For $A^{(t-1)}$, we substitute the inductive hypothesis into Eq.(\ref{eq_At_itera}).
As ${\Omega}_{s^{(\tau-1)}\to s^{(\tau)}}^{(\tau)}\geq |\tilde{\omega}_{s^{(\tau-1)}\to s^{(\tau)}}^{(\tau)}| \geq0$, by Eq.(\ref{eq_At_itera}), we have
\begin{equation}
    \begin{split}
        A^{(t-1)}_{s^{(t-1)}}
        &\geq 2^{T-t+1}\sum_{s^{(t)}}\cdots\sum_{s^{(T)}}\prod_{\tau=t}^T{\Omega}_{s^{(\tau-1)}\to s^{(\tau)}}^{(\tau)}
    -\left(2^{T-t+1} - 1\right)\sum_{s^{(t)}}\cdots\sum_{s^{(T)}}\prod_{\tau=t}^T{\Omega}_{s^{(\tau-1)}\to s^{(\tau)}}^{(\tau)}
    \\&=\sum_{s^{(t)}}\cdots\sum_{s^{(T)}}\prod_{\tau=t}^T{\Omega}_{s^{(\tau-1)}\to s^{(\tau)}}^{(\tau)}
    \end{split}
\end{equation}
and
\begin{equation}
    \begin{split}
        A^{(t-1)}_{s^{(t-1)}}
        &\leq 2^{T-t+1}\sum_{s^{(t)}}\cdots\sum_{s^{(T)}}\prod_{\tau=t}^T{\Omega}_{s^{(\tau-1)}\to s^{(\tau)}}^{(\tau)}
    +\left(2^{T-t+1} - 1\right)\sum_{s^{(t)}}\cdots\sum_{s^{(T)}}\prod_{\tau=t}^T{\Omega}_{s^{(\tau)}\to s^{(\tau)}}^{(\tau)}
    \\&=
    \left(2^{T-t+2} - 1\right)\sum_{s^{(t)}}\cdots\sum_{s^{(T)}}\prod_{\tau=t}^T{\Omega}_{s^{(\tau-1)}\to s^{(\tau)}}^{(\tau)}
    \end{split}
\end{equation}
Thus, if the bound holds for $t$, it also holds for $t-1$.
\begin{equation}
    \left(2^{T-t+2} - 1\right)\sum_{s^{(t)}}\cdots\sum_{s^{(T)}}\prod_{\tau=t}^T{\Omega}_{s^{(\tau-1)}\to s^{(\tau)}}^{(\tau)}
    \geq A^{(t-1)}_{s^{(t-1)}}\geq
    \sum_{s^{(t)}}\cdots\sum_{s^{(T)}}\prod_{\tau=t}^T{\Omega}_{s^{(\tau-1)}\to s^{(\tau)}}^{(\tau)}
\end{equation}
This completes the induction.
\end{proof}

\subsection{Explicit details of the proof}
Extract $\alpha^{(1)}_{s^{(1)}}$ from each term in the summation over all effective trajectories, we can simplify $q(\boldsymbol{\alpha})$ with $A^{(1)}$,
\begin{equation}
    \begin{split}
    q(\boldsymbol{\alpha})=&\frac{1}{Z}\exp \left( \sum_{t=1}^T \beta 2^{T-t} \sum_{\boldsymbol{\sigma}, \mathbf{s}} p(\boldsymbol{\sigma}, \mathbf{s}) x_{s_0} \prod_{\tau=1}^t (-1)^{\sigma^{(\tau)}} \alpha^{(\tau)}_{s^{(\tau)}} \right)
    \\
    =&\frac{1}{Z}
    \exp \left\{ \beta
    \sum_{s^{(0)}}\sum_{s^{(1)}}
    \alpha^{(1)}_{s^{(1)}}
        \tilde{\omega}_{s^{(0)}\to s^{(1)}}^{(1)}
        x_{s^{(0)}}
        \left[
        2^{T-1}\left(\sum_{s^{(2)}}\cdots \sum_{s^{(T)}}
        \prod_{\tau=2}^T
        {\Omega}_{s^{(\tau-1)}\to s^{(\tau)}}^{(\tau)}
        \right)
        \right.
        \right.
        \\
        &+
    \sum_{s^{(2)}}\alpha^{(2)}_{s^{(2)}}
        \tilde{\omega}_{s^{(1)}\to s^{(2)}}^{(2)}
        \left(
        2^{T-2}\left(\sum_{s^{(3)}}\cdots \sum_{s^{(T)}}
        \prod_{\tau=3}^T
        {\Omega}_{s^{(\tau-1)}\to s^{(\tau)}}^{(\tau)}
        \right)
        +\cdots
        \right.
        \\
        &+
        \left.\left.\left.
        \sum_{s^{(T-1)}}
    \alpha^{(T-1)}_{s^{(T-1)}}
        \tilde{\omega}_{s^{(T-2)}\to s^{(T-1)}}^{(T-1)}
        \sum_{s^{(T)}}\left(
        2{\Omega}_{s^{(T-1)}\to s^{(T)}}^{(T)}
        +
        \alpha^{(T)}_{s^{(T)}}
        \tilde{\omega}_{s^{(T-1)}\to s^{(T)}}^{(T)}
        \right)
        \right)\right]\right\}
    \\
    =&\frac{1}{Z}
    \exp\left(\beta
    \sum_{s^{(0)}}\sum_{s^{(1)}}
    \alpha^{(1)}_{s^{(1)}}
        \tilde{\omega}_{s^{(0)}\to s^{(1)}}^{(1)}
        x_{s^{(0)}}A^{(1)}_{s^{(1)}}
    \right)
    \\
    =&
    \prod_{s^{(1)}}\exp\left(\beta
    A^{(1)}_{s^{(1)}}\alpha^{(1)}_{s^{(1)}}
    \sum_{s^{(0)}}
        \tilde{\omega}_{s^{(0)}\to s^{(1)}}^{(1)}
        x_{s^{(0)}}
    \right)
    \end{split}
\end{equation}
Denote $\widetilde{\boldsymbol{\alpha}}=\arg\max\limits_{\boldsymbol{\alpha}}\sum_{t=1}^T  2^{T-t}\pi^{(t)}(\boldsymbol{\alpha})$.
At $\beta\to+\infty$ limit, only the post-selection criteria $\widetilde{\boldsymbol{\alpha}}$ survives, yielding that $\lim_{\beta\to+\infty}q(\widetilde{\boldsymbol{\alpha}})=1$, whereas all other ${\boldsymbol{\alpha}}$ yield that $\lim_{\beta\to+\infty}q({\boldsymbol{\alpha}}) = 0, \widetilde{\boldsymbol{\alpha}}\neq \boldsymbol{\alpha}$.

According to Lemma\ref{lm_A}, $A^{(1)}_{s^{(1)}}\geq0$, and $A^{(1)}_{s^{(1)}}=0$ iff all $\prod_{\tau=2}^T{\Omega}_{s^{(\tau-1)}\to s^{(\tau)}}^{(\tau)}=0$.
Thus, $\widetilde{\boldsymbol{\alpha}}$ must maximize each terms in the product.
Recalling that $\alpha^{(1)}_{s^{(1)}}\in\{0,1\}$, we have
\begin{equation}
    \widetilde\alpha^{(1)}_{s^{(1)}}=
    \begin{cases} 
  1, & \sum_{s^{(0)}}
        \tilde{\omega}_{s^{(0)}\to s^{(1)}}^{(1)}
        x_{s^{(0)}}\geq0 \\ 
  0,      & \sum_{s^{(0)}}
        \tilde{\omega}_{s^{(0)}\to s^{(1)}}^{(1)}
        x_{s^{(0)}}<0
\end{cases}
\label{eqs_talpha1}
\end{equation}
which further yielding that
\begin{equation}
   \begin{split}
       &\widetilde\alpha^{(1)}_{s^{(1)}}\sum_{s^{(0)}} \tilde{\omega}_{s^{(0)}\to s^{(1)}}^{(1)}x_{s^{(0)}}
   \\=&\frac{1}{C^{(0)}C^{(1)}}\text{ReLU}
   \left(\sum_{s^{(0)}}\widetilde W^{(1)}_{s^{(0)}\to s^{(1)}}\hat{x}_{s^{(0)}}\right)
   \\=&
   \frac{1}{C^{(0)}C^{(1)}}\widetilde h^{(1)}_{s^{(1)}}
   \end{split}
\end{equation}

Next, we investigate the configurations $\alpha^{(2)}_{s^{(2)}}$ that maximize the exponential terms.
By extracting the $\alpha^{(2)}_{s^{(2)}}$ from each term of $A^{(1)}_{s^{(1)}}$, we can rewrite $q(\boldsymbol{\alpha})$ as
\begin{equation}
    \begin{split}
    q(\boldsymbol{\alpha})=&\frac{1}{Z}\exp \left( \sum_{t=1}^T \beta 2^{T-t} \sum_{\boldsymbol{\sigma}, \mathbf{s}} p(\boldsymbol{\sigma}, \mathbf{s}) x_{s_0} \prod_{\tau=1}^t (-1)^{\sigma^{(\tau)}} \alpha^{(\tau)}_{s^{(\tau)}} \right)
    \\
     =&
    \frac{1}{Z}\exp \left\{ \beta
    \sum_{s^{(0)}}\sum_{s^{(1)}}
    \alpha^{(1)}_{s^{(1)}}
        \tilde{\omega}_{s^{(0)}\to s^{(1)}}^{(1)}
        x_{s^{(0)}}
        \left[
        2^{T-1}\left(\sum_{s^{(2)}}\cdots \sum_{s^{(T)}}
        \prod_{\tau=2}^T
        {\Omega}_{s^{(\tau-1)}\to s^{(\tau)}}^{(\tau)}
        \right)
        \right.
        \right.
        \\
        &+
    \sum_{s^{(2)}}\alpha^{(2)}_{s^{(2)}}
        \tilde{\omega}_{s^{(1)}\to s^{(2)}}^{(2)}
        \left(
        2^{T-2}\left(\sum_{s^{(3)}}\cdots \sum_{s^{(T)}}
        \prod_{\tau=3}^T
        {\Omega}_{s^{(\tau-1)}\to s^{(\tau)}}^{(\tau)}
        \right)
        +\cdots
        \right.
        \\
        &+
        \left.\left.\left.
        \sum_{s^{(T-1)}}
    \alpha^{(T-1)}_{s^{(T-1)}}
        \tilde{\omega}_{s^{(T-2)}\to s^{(T-1)}}^{(T-1)}
        \sum_{s^{(T)}}\left(
        2{\Omega}_{s^{(T-1)}\to s^{(T)}}^{(T)}
        +
        \alpha^{(T)}_{s^{(T)}}
        \tilde{\omega}_{s^{(T-1)}\to s^{(T)}}^{(T)}
        \right)
        \right)\right]\right\}
    \\=&
    \frac{1}{Z}\exp \left\{ \beta
    \sum_{s^{(0)}}\sum_{s^{(1)}}
    \alpha^{(1)}_{s^{(1)}}
        \tilde{\omega}_{s^{(0)}\to s^{(1)}}^{(1)}
        x_{s^{(0)}}
        \left[
        2^{T-1}\left(\sum_{s^{(2)}}\cdots \sum_{s^{(T)}}
        \prod_{\tau=2}^T
        {\Omega}_{s^{(\tau-1)}\to s^{(\tau)}}^{(\tau)}
        \right)
        \right.
        \right.
        \\
        &+\left.\left.
    \sum_{s^{(2)}}\alpha^{(2)}_{s^{(2)}}
        \tilde{\omega}_{s^{(1)}\to s^{(2)}}^{(2)}
        A^{(2)}_{s^{(2)}}
        \right]\right\}
        \\=&
   \frac{1}{Z} \exp \left\{ \beta
    \sum_{s^{(0)}}\sum_{s^{(1)}}
    \alpha^{(1)}_{s^{(1)}}
        \tilde{\omega}_{s^{(0)}\to s^{(1)}}^{(1)}
        x_{s^{(0)}}
        \left[
        2^{T-1}\left(\sum_{s^{(2)}}\cdots \sum_{s^{(T)}}
        \prod_{\tau=2}^T
        {\Omega}_{s^{(\tau-1)}\to s^{(\tau)}}^{(\tau)}
        \right)
        \right]
        \right\}
        \\
        &\cdot\exp\left\{\left(
        \beta
    \sum_{s^{(0)}}\sum_{s^{(1)}}\sum_{s^{(2)}}
    \alpha^{(1)}_{s^{(1)}}\alpha^{(2)}_{s^{(2)}}
        \tilde{\omega}_{s^{(0)}\to s^{(1)}}^{(1)}
        x_{s^{(0)}}
        \tilde{\omega}_{s^{(1)}\to s^{(2)}}^{(2)}
        A^{(2)}_{s^{(2)}}
        \right)\right\}
        \\
    \end{split}
    \label{eq_expA2}
\end{equation}
The first term in Eq.(\ref{eq_expA2}) is maximized by appropriate setting configurations $\widetilde\alpha^{(1)}_{s^{(1)}}$, as given in Eq.(\ref{eqs_talpha1}).
The second term can be further rewritten as
\begin{equation}
    \begin{split}
        &\exp\left\{\left(
        \beta
    \sum_{s^{(0)}}\sum_{s^{(1)}}\sum_{s^{(2)}}
    \alpha^{(1)}_{s^{(1)}}\alpha^{(2)}_{s^{(2)}}
        \tilde{\omega}_{s^{(0)}\to s^{(1)}}^{(1)}
        x_{s^{(0)}}
        \tilde{\omega}_{s^{(1)}\to s^{(2)}}^{(2)}
        A^{(2)}_{s^{(2)}}
        \right)\right\}
        \\=&
        \exp\left\{\left(
        \beta
    \sum_{s^{(2)}}A^{(2)}_{s^{(2)}}\alpha^{(2)}_{s^{(2)}}
    \sum_{s^{(1)}}\tilde{\omega}_{s^{(1)}\to s^{(2)}}^{(2)}
    \sum_{s^{(0)}}
    \alpha^{(1)}_{s^{(1)}}
        \tilde{\omega}_{s^{(0)}\to s^{(1)}}^{(1)}
        x_{s^{(0)}}
        \right)\right\}
        \\=&
        \exp\left\{\left(
        \beta
    \sum_{s^{(2)}}A^{(2)}_{s^{(2)}}\alpha^{(2)}_{s^{(2)}}
    \sum_{s^{(1)}}\tilde{\omega}_{s^{(1)}\to s^{(2)}}^{(2)}
    \frac{1}{C^{(1)}}h_{s^{(1)}}^{(1)}
        \right)\right\}
        \\=&
        \prod_{s^{(2)}}\exp\left\{\left(
       \frac{\beta}{C^{(1)}} 
    A^{(2)}_{s^{(2)}}\alpha^{(2)}_{s^{(2)}}
    \sum_{s^{(1)}}\tilde{\omega}_{s^{(1)}\to s^{(2)}}^{(2)}
    h_{s^{(1)}}^{(1)}
        \right)\right\}
    \end{split}
\end{equation}
Therefore, we obtain $\widetilde\alpha^{(2)}_{s^{(2)}}$ as
\begin{equation}
    \widetilde\alpha^{(2)}_{s^{(2)}}=
    \begin{cases} 
  1, & \sum_{s^{(1)}}\tilde{\omega}_{s^{(1)}\to s^{(2)}}^{(2)}
    \widetilde h_{s^{(1)}}^{(1)}\geq0 \\ 
  0,      & \sum_{s^{(1)}}\tilde{\omega}_{s^{(1)}\to s^{(2)}}^{(2)}
    \widetilde h_{s^{(1)}}^{(1)}<0
\end{cases}
\label{eq_alpha2}
\end{equation}
yielding that
\begin{equation}
   \begin{split}
&\widetilde\alpha^{(2)}_{s^{(2)}}\sum_{s^{(1)}}\tilde{\omega}_{s^{(1)}\to s^{(2)}}^{(2)}
    \widetilde\alpha^{(1)}_{s^{(1)}}\sum_{s^{(0)}} \tilde{\omega}_{s^{(0)}\to s^{(1)}}^{(1)}x_{s^{(0)}}
   \\=&\frac{1}{C^{(0)}C^{(1)}C^{(2)}}\text{ReLU}
   \left(\sum_{s^{(1)}}\widetilde W^{(2)}_{s^{(1)}\to s^{(2)}}\widetilde h_{s^{(1)}}^{(1)}\right)
   \\=&
   \frac{1}{C^{(0)}C^{(1)}C^{(2)}}\widetilde h^{(2)}_{s^{(2)}}
   \end{split}
\end{equation}
Repeating this process iteratively, we prove that $\widetilde{\boldsymbol{\alpha}}$ is exactly same as the activated ($\widetilde\alpha^{(t)}_{s^{(t)}}=1$)/inactivated ($\widetilde\alpha^{(t)}_{s^{(t)}}=0$) state of the corresponding deep ReLU network. Recalling Eq.(\ref{eq_ztalpha}), we have
\begin{equation}
    \mathbf{\widetilde z}^{(t)}(\widetilde{\boldsymbol{\alpha}})
    =\left(\prod_{\tau=0}^tC^{(\tau)}\right)^{-1}\widetilde{\mathbf h}^{(t)}
\end{equation}

Recalling that at $\beta\to+\infty$ limit, only the post-selection criteria $\widetilde{\boldsymbol{\alpha}}$ survives, yielding that $\lim_{\beta\to+\infty}q(\widetilde{\boldsymbol{\alpha}})=1$, whereas all other ${\boldsymbol{\alpha}}$ yield that $\lim_{\beta\to+\infty}q({\boldsymbol{\alpha}}) = 0, \widetilde{\boldsymbol{\alpha}}\neq \boldsymbol{\alpha}$.
Therefore $\lim_{\beta\to+\infty} \langle\mathbf{\widetilde z}^{(t)}\rangle =\mathbf{\widetilde z}^{(t)}(\widetilde{\boldsymbol{\alpha}})$, Eq.(\ref{eq_exacth}) is proved.
$g_k$ is the $k$-th component of ${\widetilde z}^{(T)}$, we have $\left(\prod_{\tau=0}^tC^{(\tau)}\right)g_k(\widetilde{\boldsymbol{\alpha}})=f_k$. 
As $\lim_{\beta\to+\infty} \langle g_k\rangle =g_k(\widetilde{\boldsymbol{\alpha}})$, and $\lim_{\beta\to+\infty} \langle g_k\rangle^2= \lim_{\beta\to+\infty} \langle (g_k)^2\rangle =(g_k(\widetilde{\boldsymbol{\alpha}}))^2$,  Eq.(\ref{eq_exact_loss}) is proved.

\end{document}